%% file: editDistance.tex
\newif\ifprocs
\procstrue
\procsfalse %

\ifprocs
\else
\documentclass[11pt]{article}
\usepackage{fullpage}
\fi

\ifprocs
\else
\usepackage{amsmath}
\fi
\usepackage{amsthm}
\usepackage{amssymb}
\usepackage{enumerate}
\usepackage{mathpazo}
\usepackage{tikz}
\usepackage{nicefrac}
\usepackage{rotating}
\usepackage{setspace}
\usepackage{relsize}
\usepackage{color}
\ifprocs
\else
\usepackage{hyperref}
\hypersetup{colorlinks=true,allcolors=blue}
\fi
\usepackage{bbm}

\usepackage{thmtools}
\newtheorem{theorem}{Theorem}[section]

\newtheorem{corollary}[theorem]{Corollary}
\newtheorem{lemma}[theorem]{Lemma}

\newtheorem{claim}[theorem]{Claim}

\newtheorem*{theorem*}{Theorem}
\newtheorem*{proposition*}{Proposition}
\newtheorem*{lemma*}{Lemma}
\newtheorem{definition}[theorem]{Definition}
\newtheorem{remark}[theorem]{Remark}

\usepackage{dsfont}

\newcommand{\indic}{\mathds{1}}
\newcommand{\indicator}[1]{\indic_{#1}} 
\input{macros}

\SetKwInOut{Input}{Input}
\SetKwInOut{Output}{Output}

\begin{document}

\ifprocs
\else
\title{Sublinear Algorithms for Gap Edit Distance%
  \thanks{Part of this work was done while the authors were visiting
    the Simons Institute for the Theory of Computing.
  }}
\author{
  Elazar Goldenberg%
  \thanks{The Academic College of Tel Aviv-Yaffo. 
    Email: \texttt{elazargo@mta.ac.il}
  }
  \and
  Robert Krauthgamer%
  \thanks{Weizmann Institute of Science.
    Work partially supported by ONR Award N00014-18-1-2364,
    the Israel Science Foundation grant \#1086/18,
    and a Minerva Foundation grant.
    Email: \texttt{robert.krauthgamer@weizmann.ac.il}
  }
  \and
  Barna Saha%
  \thanks{University of California Berkeley. Work partially supported by a 
  NSF CAREER Award 1652303, and the Alfred P. Sloan Fellowship.
    Email: \texttt{barnas@berkeley.edu}
  }
}
\fi

\maketitle

\begin{abstract}
The edit distance is a way of quantifying how similar two strings are to one another by counting the minimum number of character insertions, deletions, and substitutions required to transform one string into the other.
A simple dynamic programming computes the edit distance between two strings of length $n$ in $O(n^2)$ time,
and a more sophisticated algorithm runs in time $O(n+t^2)$
when the edit distance is $t$ [Landau, Myers and Schmidt, SICOMP 1998]. 
In pursuit of obtaining faster running time, the last couple of decades have seen a flurry of research on approximating edit distance,
including polylogarithmic approximation in near-linear time
[Andoni, Krauthgamer and Onak, FOCS 2010],
and a constant-factor approximation in subquadratic time
[Chakrabarty, Das, Goldenberg, Kouck\'y and Saks, FOCS 2018].

We study sublinear-time algorithms for small edit distance,
which was investigated extensively because of its numerous applications. 
Our main result is an algorithm for distinguishing
whether the edit distance is at most $t$ or at least $t^2$
(the quadratic gap problem) in time $\tilde{O}(\frac{n}{t}+t^3)$.
This time bound is sublinear roughly for all $t$ in $[\omega(1), o(n^{1/3})]$,
which was not known before. 
The best previous algorithms solve this problem in sublinear time
only for $t=\omega(n^{1/3})$ [Andoni and Onak, STOC 2009].

Our algorithm is based on a new approach that adaptively switches
between uniform sampling and reading contiguous blocks of the input strings.
In contrast, all previous algorithms choose which coordinates to query non-adaptively. 
Moreover, it can be extended to solve the $t$ vs $t^{2-\epsilon}$ gap problem
in time $\tilde{O}(\frac{n}{t^{1-\epsilon}}+t^3)$.
\end{abstract}

\ifprocs
\else
\newpage
\fi

\section{Introduction}

The \emph{edit distance} (aka \emph{Levenshtein distance})~\cite{Lev65} is a widely used distance measure between pairs of strings $x,y $ over some alphabet $\Sigma$. It finds applications in several fields like computational biology, pattern recognition, text processing, information retrieval and many more. The edit distance between $x$ and $y$, denoted by  $\ed(x,y)$, is defined as the minimum number of character insertions, deletions, and substitutions needed for converting $x$ into $y$. Due to its immense applicability, the computational problem of computing the edit distance between two given strings $x$ and $y \in \Sigma^n$ is of prime interest to researchers in various domains of computer science. A simple dynamic program solves this problem in time $O(n^2)$.  Moreover, assuming the strong exponential time hypothesis (SETH), there does not exist any truly subquadratic algorithm for  computing the edit distance~\cite{BI15,ABW15,BK15,AHWW16}. 
	
For many applications where the data is very large, a quadratic running time is prohibitive, and it is highly desirable to design faster algorithms, 
even approximate ones that compute a near-optimal solution. 
The last couple of decades have seen exciting developments in this frontier.
In time $\tilde{O}(n^{1+\epsilon})$ for arbitrary $\epsilon >0$, 
it is now possible to approximate the edit distance within factor  $O(\log^{O(\frac{1}{\epsilon})}{n})$ \cite{AKO10}.
This bound was a culmination of earlier results where the approximation bound improved from $O(\sqrt{n})$ in linear time \cite{LMS98}
to $O(n^{3/7})$ and $n^{1/3+o(1)}$ in quasi-linear time \cite{BJKK04,BES06},
to $2^{O(\sqrt{\log{n}\log{\log{n}}})}$
in time $O(n 2^{O(\sqrt{\log{n}\log{\log{n}}})})$ \cite{AO09}.
Recently, a breakthrough by Chakrabarty, Das, Goldenberg, Kouck\'y and Saks~\cite{CDGKS18} obtained the first constant factor approximation algorithm for computing edit distance with a subquadratic running time. However, when restricted to strictly linear time algorithms, a $\sqrt{n}$ approximation still remains the best possible \cite{LMS98,Saha14,CGK16}. In fact, when $\ed(x,y)=t$, the algorithm by Landau, Myers and Schmidt runs in $O(n+t^2)$ time \cite{LMS98}. Thus for $t \leq \sqrt{n}$, the edit distance can be computed exactly in linear time. This algorithm has found wide-spread applications \cite{AG98,Navarro98,BVKB10,CGRS16} and is also known to be optimal under SETH. 

\paragraph{Sublinear time}
Following this quest for ever faster algorithms,
it is natural to seek sublinear-time approximation algorithms.
We study the regime of small edit distance $t$,
which was investigated extensively in the literature
because of its high relevance to many applications. 
In computational biology, for example, it is often only necessary to compare genomic sequences that are highly similar and quickly get rid of sequences that are far apart, e.g., some sequencing projects target a strain or species that is closely related to an already-sequenced organism~\cite{DKFPWS99}.
A major difficulty is that genomic sequences are comprised of highly repetitive patterns (repeats) whose frequency and placement contain important information about genetic variation, gene regulation, human disease condition, etc.\ \cite{TS12,SA10}.
In a text corpora, detecting plagiarism and eliminating duplicates require identification of document pairs that are small edit distance apart.
These applications can benefit from super-fast algorithms that answer whether the edit distance is below a threshold $t$ or above $f(t)$ for some function $f$, known as the \emph{gap edit distance} problem. The goal here is to design algorithms that are simultaneously highly efficient and have $f(t)$ as close to $t$ as possible.

\paragraph{What is the right gap?}
We focus on $f(t)=t^2$, i.e., a quadratic gap as our main test case.  
This is perhaps the most natural choice other than $f(t)=\Theta(t)$
(i.e., multiplicative approximation),
and is also motivated by known results for linear and sublinear time.
\begin{itemize} %
  \compactify

\item In linear time, the algorithm of~\cite{LMS98} can solve the $t$ vs $t^2$ gap problem. So far, no linear-time algorithm is known to beat this bound \cite{Saha14,CGK16}. Bar-Yossef, Jayram, Krauthgamer and Kumar~\cite{BJKK04}  introduced the term \emph{gap edit distance} and solved the $t$ vs $t^2$ gap problem for non-repetitive strings. Their algorithm computes a constant-size sketch but still requires a linear pass over the data. This result was later improved to hold for general strings \cite{CGK16} via embedding into Hamming distance, but again in linear time.

\item The study of sublinear-time algorithms for edit distance was initiated by Batu, Ergun, Kilian, Magen, Raskhodnikova, Rubinfeld and Sami~\cite{BEKMRRS03},
  who designed an algorithm for the $t$ vs $\Omega(n)$ gap problem,
  thereby solving the quadratic gap problem only for $t=\Omega(\sqrt{n})$.
  Currently, the best sublinear-time algorithm, by Andoni and Onak~\cite{AO09}, solves the $t$ vs $t^2$ gap problem for all $t=\omega(n^{1/3})$. For $t=n^{1/3+\epsilon}$, their running time is $n^{1-3\epsilon+o(1)}$.
  Solving the quadratic gap problem appears to become harder as $t$ gets  smaller
  because locating fewer edit operations will require more queries,
  and the approximation factor $t$ gets smaller.   
  (This is in contrast to the time bound $O(n+t^2)$ of \cite{LMS98}.)
\end{itemize}

\subsection{Results}
We design a sublinear time algorithm for the $t$ vs $f(t)=t^2$ gap problem.
Its running time is $\tilde{O}(\frac{n}{t}+t^3)$,
which is indeed sublinear for all $t \in [\tilde{\omega}(1), o(n^{1/3})]$.%
\footnote{Throughout,
  the tilde notation $\tO(\cdot)$ and $\tilde{\omega}(\cdot)$
  hide factors that are polylogarithmic in $n$. 
}

\begin{theorem}\label{thm:main}
  There exists an algorithm that, 
  given as input strings $x,y\in \zo^n$ and an integer $t \le \sqrt{n}$,
  has query and time complexity bounded by $O(\frac{n\log n}{t}+t^3)$,
  and satisfies the following:
  \begin{itemize}
  \item If $\ed(x,y)\le t/2$ it outputs \close with probability $1$.
  \item If $\ed(x,y)> 13t^2$ it outputs \far with probability at least  $2/3$.
  \end{itemize}
\end{theorem}
Therefore, coupled with the result of \cite{AO09}, we get sublinear time-complexity for the quadratic gap problem for $t \in [\tilde{\omega}(1),o(n^{1/3})]\cup[\omega(n^{1/3}),n]$. This leaves a very interesting open question as to what happens when $t=\Theta(n^{1/3})$.

Our algorithm has two more nice features. First, sometimes one also requires that the algorithm finds an \emph{alignment} of two strings: $x$ and $y$, i.e., a series of edit operations that transform $x$ into $y$. Our algorithm can succinctly represent an alignment in $\tilde{O}(t^2)$ bits even though it runs in sublinear time.
Second, the algorithm can be easily extended to solve the $t$ vs $f(t)=t^{2-\epsilon}$ gap problem by paying slightly higher in the running time/query complexity: $\tilde{O}(\frac{n}{t^{1-\epsilon}}+t^3)$.

\paragraph{Previous Work}
Batu et al.'s algorithm distinguishes $t=n^\alpha$ vs $f(t)=\Omega(n)$ in $O(n^{\maxx{2\alpha-1,\alpha/2}})$ time for any fixed $\alpha > 1$ \cite{BEKMRRS03}. Their approach crucially depends on $f(t)=\Omega(n)$ and cannot distinguish between (say) $n^{0.1}$ and $n^{0.99}$.
The best sublinear-time algorithm known for gap edit distance,
by Andoni and Onak \cite{AO09}, 
distinguishes between $t=n^\alpha$ vs $f(t)=n^\beta$ for $\beta > \alpha$
in time $O(n^{2+\alpha-2\beta+o(1)})$.
For the quadratic gap problem, i.e., $\beta=2\alpha$,
this time bound is $O(n^{2-3\alpha+o(1)})$,
which becomes worse as $t$ gets smaller (as discussed earlier).
For example, when $t=n^{1/4}$, the known algorithm is not sublinear,
whereas ours runs in time $\tilde{O}(n^{3/4})$.

Presence of repeated patterns make the gap edit distance problem significantly difficult to approximate. When no repetition is allowed, the state-of-the-art sublinear-time algorithms of~\cite{AN10} for the Ulam metric
(edit distance with no repetition, which obviously requires a large alphabet) 
distinguish between $t$ vs $\Theta(t)$ in $O(\frac{n}{t}+\sqrt{n})$ time,
achieving a bound that is similar to the folklore sampling algorithm for approximating Hamming distance. 
There is a long line of work on edit distance and related problems,
aiming to achieve fast running time \cite{AN10,AKH13,Saha17,BEGHS18,HSSS19-a}, low distortion embedding \cite{OR07,KR06,CGK16,BZ16}, small space complexity \cite{CGK16,BZ16,BJKK04} and parallel algorithms \cite{HSS19-b}. The work of Andoni, Onak and Krauthgamer~\cite{AKO10} achieves a sublinear asymmetric query complexity for approximating edit distance; however it does not lead to any sublinear time algorithm since one of the strings must be read in its entirety.

\subsection{Techniques}	

As a warmup, we start with a simple algorithm that has asymmetric query complexity -- it queries $\tilde{O}(\frac{n}{t})$ positions in $x$,
but may query the entire string $y$.
This is comparable to the Hamming metric,
where simply querying $\tilde{O}(\frac{n}{t})$ positions uniformly at random
in $x$ and the same positions in $y$, suffice to solve $t$ vs $\Theta(t)$ gap.
However, this simple uniform sampling fails miserably to estimate edit distance,
even when there is a single character insertion or deletion.
Our simple algorithm reads $\tilde{O}(\frac{n}{t})$ random positions in $x$,
but since $x_i$ might be matched to any $y_{i+d}$, $d \in [-t..+t]$, 
our algorithm has to read the entire string $y$.
(In this outline, we call $d$ a shift, and later call it a diagonal.)

Even when the entire string $y$ is known, 
we cannot hope that this approach distinguishes better than $t$ vs $f(t)=t^2$.
To see this, consider two scenarios.
In one scenario, $y$ is obtained from $x$ via $t^2$ substitutions.
Since the algorithm samples $x$ at a rate of $~\frac{1}{t}$,
we expect to see about $t$ of these substitutions.
In an alternative scenario, $x$ is partitioned into $t$ substrings
of length $\frac{n}{t}$, and $y$ is obtained from $x$
by a circular shift by one position of each of the $t$ parts (substrings).
Now, the edit distance between $x$ and $y$ is at most $2t$,
and assuming the sample of $x$ contains at least one symbol 
 from each part of $x$, the best alignment of the sampled $x$ with $y$ 
 will still constitute of $O(t)$ insertions/deletions.
These two cases will be indistinguishable to an algorithm that
aligns the samples in $x$ with the string $y$,
and thus the best separation possible in this approach is $t$ vs $f(t)=t^2$. 

To avoid sampling the entire string $y$,
one may need to sample $x$ at a lower rate 
or to sample $x$ non-uniformly in \emph{contiguous positions} (blocks).
In the former case, the separation between $t$ and $f(t)$ will only increase. In the latter case, an algorithm that samples (say) $\frac{n}{t^2}$ blocks of length $O(t)$ in $x$ can be shown to solve only a $t$ vs $t^2$ gap even for Hamming distance, and for edit distance we will need $f(t)=t^3$. 

In order to overcome these barriers, we employ both contiguous sampling and uniform sampling together, and in fact switch between them \emph{adaptively}. 
The contiguous sampling suggests plausible shifts that a low-cost alignment
may use. These plausible shifts are then checked probabilistically through uniform sampling. However, if we need to check every plausible shift via uniform sampling, the query (and time) complexity will again become linear.
A technical observation based on \cite{CGK16b} helps us here --- if two substrings can be matched under two distinct shifts $d$ and $d'$, then the substrings must have a repeated pattern. In other words, the substrings are periodic with a pattern of length $|d-d'|$.
The crux is that instead of checking each shift individually, 
we instead check for this repeated pattern via uniform sampling. When we witness a deviation from the periodicity (e.g., change in pattern),
we execute a fast test to identify all shifts that ``see'' a mismatch
(and increase our estimate of their cost). We alternate between the non-uniform and uniform sampling at an appropriate rate to achieve the desired query complexity and the running time.

In contrast, all previous sublinear/sampling algorithms,
including~\cite{BEKMRRS03,AO09,AKO10,AN10}
choose which coordinates to query non-adaptively.

\paragraph{Organization}
Section~\ref{sec:prelims} lays the groundwork for our main algorithm. 
It starts by introducing (in Section~\ref{sec:GridGraph})
the concept of a grid graph,
which represents the edit distance as a shortest-path computation in a graph.
It then describes (in Section~\ref{sec:GridSampling}) the uniform sampling technique,
which can be viewed as sampling of the grid graph,
leading to a simple algorithm with asymmetric query complexity.

Section~\ref{sec:subLinear} presents our main result.
It starts with a method (in Section~\ref{sec:LMS})
that computes a shortest path using a more selective scan of the grid graph.
It then describes (in Section~\ref{sec:AlgorithmDescription})
our main algorithm,
which combines the aforementioned techniques of sampling the grid graph
and of scanning it more selectively.

\section{Preliminaries: The Grid Graph and Uniform Sampling}
\label{sec:prelims}

\paragraph{Notation}
Let $x\in \Sigma^n$ be a string of length $n$ over alphabet $\Sigma$.
For a set $S\subseteq [n]$,
we denote by $x_S$ the restriction of $x$ to positions in $S$
(in effect, we treat $S$ as if it is ordered in the natural order). 
Oftentimes, $S$ is contiguous (i.e., an interval)
and then $x_S$ is a substring of $x$.
For $d\in \pmn$ and a set $S\subseteq[n]$ we define
$S+d \eqdef \sett{s+d}{s\in S }$.
As usual, $[i..j]$ denotes $\set{i,\ldots,j}$ for integers $i,j$. 

A string $x\in \Sigma ^n$ is called \emph{periodic}
with \emph{period length} $m<n$ and \emph{period pattern} $p\in \Sigma^m$
if $x= p^{\lfloor n/m \rfloor} \circ q$, where $q=p_{[1.. (n\bmod p)]}$ and $\circ$ means concatenation of strings.
Here and throughout we assume that $(n\bmod p)$
returns a value in the range $[1..p]$ (rather than $[0..p-1]$ as usual).

\subsection{Edit Distance as a Shortest Path in a Grid Graph}
\label{sec:GridGraph}

Given an input $x,y\in \zo^n$ to the edit distance problem,
 it is natural to consider the following directed graph $G_{x,y}$,
which we refer to as the grid graph.
It has vertex set $[0..n]\times \pmn$,
The graph has the following weighted edges 
(provided both endpoints are indeed vertices): 
\begin{enumerate}[(i)] \compactify
	\item Deletion edges: $(i ,d) \to (i+1,d-1)$ with weight $1$  corresponding to a character deletion.
	\item Insertion edges: $(i,d) \to (i,d+1)$ with weight $1$  corresponding to character insertions.
	\item Matching/substitution edges: $(i, d) \to (i+1,d)$ with weight either $0$ or $1$ depending on whether $x_{i+1}= y_{i+d+1}$ or not. 
          Such an edge corresponds to a character match/substitution.
\end{enumerate}
See Figure~\ref{fig:editgraph} for illustration.
Throughout,
we call $i$ the \emph{row} and $d$ the \emph{diagonal} of a vertex $(i,d)$,
and refer to weight also as cost.

\ifprocs
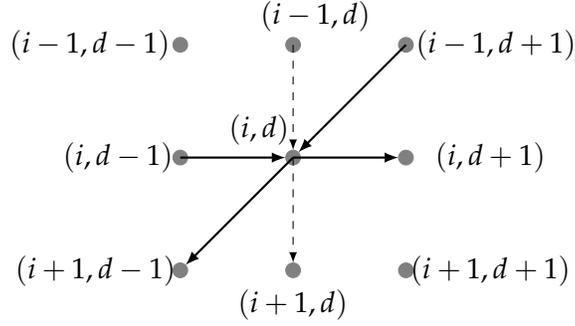
\begin{figure*}
\else
\begin{figure}[ht]
\fi
	\centering
	\begin{tikzpicture}[scale=1.5,shorten >=1mm,>=latex]
	\tikzstyle gridlines=[color=black!20,very thin]

	\draw[fill,color=black!50] (1,1) circle (0.65mm);
	\node at (0.7,1.25) {$(i,d)$};
	
	\draw[fill,color=black!50] (2,2) circle (0.65mm);
	\node at (2.8,2) {$(i-1,d+1)$};

	\draw[->,thick] (2,2)--(1,1);
	
	\draw[fill,color=black!50] (0,1) circle (0.65mm);
	\node at (-0.55,1) {$(i,d-1)$};
	
	\draw[fill,color=black!50] (0,2) circle (0.65mm);
	\node at (-0.8,2) {$(i-1,d-1)$};
	
	\draw[->,thick] (0,1)--(1,1);

	\draw[fill,color=black!50] (1,2) circle (0.65mm);
	\node at (1.2,2.25) {$(i-1,d)$};
	
	\draw[->,dashed] (1,2)--(1,1);

	\draw[fill,color=black!50] (2,1) circle (0.65mm);
	\node at (2.75,1) {$(i,d+1)$};
	
	\draw[fill,color=black!50] (2,0) circle (0.65mm);
	\node at (2.75,0) {$(i+1,d+1)$};

	\draw[->,thick] (1,1)--(2,1);

	\draw[fill,color=black!50] (1,0) circle (0.65mm);
	\node at (1,-0.3) {$(i+1,d)$};
	
	\draw[->,dashed] (1,1)--(1,0);

	\draw[fill,color=black!50] (0,0) circle (0.65mm);
	\node at (-0.75,0) {$(i+1,d-1)$};
	
	\draw[->,thick] (1,1)--(0,0);
	
	\end{tikzpicture}
	\caption{A typical vertex in $G_{x,y}$ has $3$ incoming and $3$ outgoing edges. Thick edges have cost $1$ corresponding to deletion/insertion, dashed edges have cost $0/1$ corresponding to substitution.}
	\label{fig:editgraph}
	\vspace{0.1cm}
        \hrule
	\vspace{-0.1cm}
\ifprocs
\end{figure*}
\else
\end{figure}
\fi

We assign to each vertex $(i,d)$ a cost $c(i,d)$
according to the shortest path (meaning of minimum total weight) 
from $(0,0)$ to $(i,d)$.
The following two lemmas can be easily proven using standard dynamic programming arguments; we omit the details.

\begin{lemma}
  One can compute the cost of every vertex in the grid graph
  by scanning the vertices in row order (i.e., from $0$ to $n$)
  and inside each row scanning the diagonals in order from $-n$ to $n$.
  Moreover, computing the cost of each vertex $(i,d)$ requires inspecting
  only 3 earlier vertices, namely, $(i-1,d)$, $(i-1,d+1)$ and $(i,d-1)$.
\end{lemma}

\begin{lemma} \label{lem:gridcost}
  There is a one-to-one correspondence between paths from $(0,0 )$ to $(i,d)$
  to alignments from $x_{[1.. i]}$ into $y_{[1.. i+d]}$
  (an alignment is a set of character deletions, insertions and substitutions that converts one string to the other). 
  Moreover, each path's weight is equal to the corresponding alignment's cost,
  and thus
  $$c(i,d) = \ed(x_{[1..i]}, y_{[1..i+d]}).$$ 
\end{lemma}

Observe that when the edit distance is bounded by a parameter $t$,
the optimal path goes only through vertices in $[0..n]\times [-t..t]$, 
and thus the algorithm can be restricted to this range.
We sometimes refer to the two vertices $(0,0)$ and $(n,0)$
as the \emph{source} and \emph{sink}, respectively.

\subsection{Uniform Sampling of the Rows (and Asymmetric Query Complexity)}
\label{sec:GridSampling}

As a warm-up, we now describe a simple randomized algorithm that,
given as input two strings $x,y\in \zo^n$ and a parameter $t\leq \sqrt n$, distinguishes (with high probability) whether $\ed(x,y)\le t$
or $\ed(x,y)= \Omega(t^2)$. The algorithm has asymmetric query complexity:
it queries $x$ at a rate of $\frac{\log n}{t}$ but may query $y$ in its entirety.

This algorithm is based on a sampled version of the grid graph,
denoted $G_S$, constructed as follows.
First, pick a random set $S\subseteq [n]$ where each row $i\in [n]$
is included in $S$ independently with probability $\frac{\log n}{t}$,
and add also row $0$ to $S$.
Let $s=|S|$ and denote the rows in $S$ by $0=i_1<\dots<i_s$. 
Now let the vertex set of $G_S$ be $S\times \pmt \cup \set {(n,0)}$,
and connect each vertex $(i_j,d)$, for $i_j\in S$ and $d\in \pmt$, 
to the following set of vertices (provided they are indeed vertices): 
(i) $(i_{j+1},d)$ with weight $\indicator{\set{x_{i_{j}+1}\neq y_{i_{j}+d+1}}}$;
and
(ii) $(i_{j},d+1)$ with weight $1$ 
and (iii) $(i_{j+1},d-1)$ with weight $1$. 
Finally, connect each vertex $(i_s,d)$ to the sink $(n,0)$ with
an edge of weight $\abs{d}$.

The algorithm constructs $G_S$
and then computes the shortest path from $(0,0)$ to $(n,0)$.
If its cost is at most $t$ 
then the algorithm outputs \close,
otherwise it outputs \far. 

\begin{lemma}\label{lem:gridGraph}
  For all $x,y\in \Sigma^n$, 
  \begin{itemize} \compactify
  \item
    if $\ed(x,y)\le t$ then with probability $1$
  the algorithm outputs \close; 
\item
  if $\ed(x,y)>6 t^2$ then 
  with probability at least $2/3$ the algorithm outputs \far.
\end{itemize}
\end{lemma}
 
We sketch here the proof, 
deferring details to Appendix~\ref{app:gridGarphProof}.

\paragraph{Close case $\ed(x,y)\le t$}
It suffices to show that for every source-to-sink path $\tau$ in $G_{x,y}$ 
(the original grid graph) there is in $G_S$ a corresponding source-to-sink path $\tau_S$ that has the same or lower cost.
To do this, start by letting $\tau_S$ visit the same set of vertices
as $\tau$ visits in row $0$,
starting of course at the source $(0,0)$, %
and then extend $\tau_S$ iteratively from row $i_j$ to $i_{j+1}$, as follows. 
Denote the last vertex visited by $\tau_S$ on row $i_j$ by $(i_j,d_S)$, 
and the vertices visited by $\tau$ on row $i_{j+1}$ by $(i_{j+1},d)\dots, (i_{j+1},d+\ell)$.
Now if $d_S\le d+\ell$, then extend $\tau_S$ by appending $(i_{j+1},d_S) \dots, (i_{j+1},d+\ell)$. 
Otherwise, extend it by appending $(i_{j+1},d-1)$.
Finally, after $\tau_S$ visits row $i_S$,
extend it by appending the sink $(n,0)$.

Denote by $c_{G_{x,y}}(\cdot)$ the cost of a path in $G_{x,y}$,
and by $c_{G_S}(\cdot)$ the cost of a path in $G_S$. 
We can then prove that $c_{G_S}(\tau_S) \le c_{G_{x,y}}(\tau)$,
see Claim~\ref{claim:closeCaseAnalysis} for details.

\paragraph{Far case $\ed(x,y) \ge 6t^2$}
We need the next claim,
which follows easily by the independence in sampling rows to $S$.
Let $\ham(\cdot,\cdot)$ denote the Hamming distance between two strings. 

\begin{claim}\label{claim:badSamplingEvent}
	Fix $x,y \in \Sigma^n$, $i\in [n]$ and $d\in \pmt$.
	Let $i' \ge i$ be the minimum such that
	$$ \ham (x_{[i..i']}, y_{[i..i']+d})\ge 3t,$$
	and let $H=\sett{j\in [i..i']}{x_j\neq y_{j+d}}$ be the corresponding set of Hamming errors.  
	If no such $i'$ exists then set $i'=\infty$.
	
	Define $B(i,d)$ to be the event that $i'<\infty$
	and that no row from $H$ is sampled, i.e., $H\cap S=\emptyset$.
	Then
	$$\Pr [B(i,d)]
	\le \left( 1-\frac{\log n}{t} \right)^{3t} 
	< \frac{1}{3n(2t+1)}.
	$$
\end{claim}
\sloppy
By a union bound over the set of all possible rows $i$ and diagonals $d$, 
we get that except with probability $n(2t+1) \frac 1 {3n(2t+1)}\le \frac 1 3$, none of the events $B(i,d)$ happens.
We conclude the proof by showing that whenever this happens, 
every path in $G_S$ from the source $(0,0)$ to the sink $(n,0)$
has cost strictly larger than $t$, 
and therefore our algorithm outputs \far.

\subsection{Generalized Grid Graph}
\label{sec:GenGridGraph}

It is instructive, and in fact needed for the algorithm
we describe in Section~\ref{sec:AlgorithmDescription}, 
to consider the following generalization of $G_S$ (and of $G_{xy}$). 
A \emph{generalized grid graph} has the same vertices and edges as $G_S$
(which was defined in Section~\ref{sec:GridSampling}),
except that the edges of type (iii) have arbitrary weights from the domain $\zo$.
This is in contrast to $G_S$, where all these weights
are derived from the two strings $x,y$ and thus have various correlations, 
e.g., a single $x_{i_j}$ affects the weight of many edges.
The next lemma shows that such graphs have a lot of structure. 

\begin{lemma}
Consider a generalized grid graph
and denote its rows sequentially by $0,1,2,\ldots,\card{S}-1$.
Then the cost difference between a vertex $(i,d)$ and its in-neighbors
is bounded by:
\begin{align}
  0  & \le c(i,d) - c(i-1,d) \le 1  \label{eq:mono}\\
  -1 & \le c(i,d) - c(i,d-1) \leq 1 \label{eq:bdd_diff1} \\
  -1 & \leq c(i,d) - c(i-1,d+1) \leq 1 \label{eq:bdd_diff2}
\end{align}
\end{lemma}

\begin{proof}
The three upper bounds are immediate from the triangle inequality,
hence we only need to prove the three lower bounds.

We prove the lower bound in~\eqref{eq:mono}
by induction on the grid vertices $(i,d)$ in lexicographic order
(i.e., their row is the primary key and their diagonal is secondary). 
For the inductive step,
consider $(i,d)$ and assume the lower bound holds for all previous vertices. 
The cost of $(i,d)$ is the minimum of three values coming 
from its in-neighbors, and at least one of these values is tight.
We thus have three cases. 
In the first one, the value coming from in-neighbor $(i-1,d)$ is tight,
then $c(i,d) - c(i-1,d) \in \set{0,1}$ and we are done. 
The second case is when the value coming from in-neighbor $(i,d-1)$ is tight.
Using this fact, applying the induction hypothesis to $(i,d-1)$,
and then the upper bound in~\eqref{eq:bdd_diff1}, we have
\[
  c(i,d)
  = c(i,d-1) +1
  \ge c(i-1,d-1) +1
  \ge c(i-1,d) ,
\]
as required. 
The third case, where the value coming from in-neighbor $(i-1,d+1)$ is tight,
is proved similarly, and this concludes the proof of~\eqref{eq:mono}. 

The lower bound in~\eqref{eq:bdd_diff1} follows easily
by using the upper bound in~\eqref{eq:bdd_diff2}
and then the monotonicity property~\eqref{eq:mono},
we indeed obtain 
$c(i,d-1) \le c(i-1,d) + 1 \le c(i,d) + 1$.
The lower bound in~\eqref{eq:bdd_diff2} follows similarly
$c(i-1,d+1) \leq c(i-1,d) + 1 \leq c(i,d) + 1$. 
\end{proof}

\section{Sublinear Algorithm for Quadratic Gap} 
\label{sec:subLinear}

In this section we present our main sublinear algorithm,
which combines two techniques.
The first one, explained in Section~\ref{sec:prelims}, 
is to sample uniformly rows in the grid graph $G_{x,y}$ %
and compute a shortest path in the resulting (sampled) graph $G_S$. 
The second technique scans the grid graph more selectively
in the sense of skipping some vertices in an adaptive manner. 
While this technique is known, e.g., from \cite{LMS98},
our version (presented in Section~\ref{sec:LMS}) differs from previous work 
because it scans the graph row by row. 
We then explain (in Section~\ref{sec:modetransition})
our main technical insight that allows to adaptively switch
between the two aforementioned techniques,
which correspond to uniform sampling and reading contiguous blocks (from $x,y$).
We are then ready to present the algorithm itself
(in Section~\ref{sec:AlgorithmDescription}),
followed by an analysis of its correctness and time/query complexity
(in Section~\ref{sec:analysis}),
and a discussion of some extensions.

\subsection{Selective Scan of the Grid Graph}
\label{sec:LMS}

We will make use of a technique developed in~\cite{UKK85,LMS98,LV88,Myers86}
that scans the grid graph $G_{x,y}$ more selectively,
and yet is guaranteed to compute a shortest path from $(0,0)$ to $(n,0)$. 
By itself, this technique does not yield asymptotic improvement
over a naive scan of all the vertices,
however it is crucial to our actual algorithm 
(and also to the algorithm of~\cite{LMS98},
which uses a variant of this technique
where the grid graph is scanned in ``waves'' rather than row by row). 
Along the way, we introduce three notions (dominated, potent and active)
that may apply to a diagonal $d$ at row $i$, i.e., to a vertex $(i,d)$.
To simplify the exposition, we do not discuss all the boundary cases,
and objects that do not exist (like row $-1$ or character $x_{n+1}$)
should be ignored (e.g., omitted from a minimization formula).

\paragraph{Dominated vertices}
Let $(i,d)$ be a vertex in $G_{x,y}$ of cost $h=c(i,d)$. 
If any of its in-neighbors $(i,d-1)$ and $(i-1,d+1)$ has cost $h-1$
then we say that $(i,d)$ is \emph{dominated} by that in-neighbor. 

The following observation may allow us to ``skip'' 
a dominated vertex when computing a shortest path.
Suppose that $(i,d)$ is dominated by $(i,d-1)$,
see Figure~\ref{fig:domination} for illustration
(when dominated by $(i-1,d+1)$, an analogous argument applies).
If diagonal $d-1$ \emph{has a match} at row $i+1$,
defined as $x_{i+1}=y_{i+1+(d-1)}$, 
then there exists a shortest path to $(i+1,d)$ that avoids vertex $(i,d)$, 
by going for example through $(i,d-1) \to (i+1,d-1) \to (i+1,d)$.
Notice that vertex $(i+1,d)$ must be dominated too. %
If, however, diagonal $d-1$ \emph{has a mismatch} at row $i+1$,
defined as $x_{i+1}\neq y_{i+1+(d-1)}$,
then it might be that every shortest path to $(i+1,d)$ passes through $(i,d)$. 
Notice that now $(i+1,d)$ may or may not be dominated.

\ifprocs
\begin{figure*}
\else
\begin{figure}[t]
\fi
	\begin{tikzpicture}[scale=1.5,shorten >=1mm,>=latex]
	\tikzstyle gridlines=[color=black!20,very thin]

	\draw[fill,color=black!50] (1,1) circle (0.65mm);
	\node[color=blue] at (1.75,1) {$c(i,d)=h$};
	
	\draw[fill,color=black!50] (0,1) circle (0.65mm);

	\node[color=blue] at (-1.1,1) {$c(i,d-1)=h-1$};
	
	\draw[->,dashed,thick] (0,1)--(1,1);

	\draw[fill,color=black!50] (1,0) circle (0.65mm);
	\node[color=blue] at (1.9,0) {$c(i+1,d)=h$};
	
	\draw[->,dashed,thick] (1,1)--(1,0);
	
	\node[color=blue] at (1.25,0.5) {$0/1$};
	
	\draw[fill,color=black!50] (0,0) circle (0.65mm);

	\node[color=blue] at (-1.5,0) {$c(i+1,d-1)=h-1$};
	
	\draw[->,thick] (0,1)--(0,0);
	
	\node[color=blue] at (0.2,0.5) {$0$};
	
	\draw[->,thick] (0,0)--(1,0);

	\draw[fill,color=black!50] (7,1) circle (0.65mm);
	\node[color=blue] at (7.75,1) {$c(i,d)=h$};
	
	\draw[fill,color=black!50] (6,1) circle (0.65mm);

	\node[color=blue] at (4.9,1) {$c(i,d-1)=h-1$};
	
	\draw[->,thick] (6,1)--(7,1);

	\draw[fill,color=black!50] (7,0) circle (0.65mm);
	\node[color=blue] at (8,0) {$c(i+1,d)=h$};
	
	\draw[->,thick] (7,1)--(7,0);
	
	\node[color=blue] at (7.25,0.5) {$0$};
	
	\draw[fill,color=black!50] (6,0) circle (0.65mm);

	\node[color=blue] at (4.8,0) {$c(i+1,d-1)=h$};
	
	\draw[->,dashed,thick] (6,1)--(6,0);
	
	\node[color=blue] at (6.2,0.5) {$1$};
	
	\draw[->,dashed,thick] (6,0)--(7,0);

	\end{tikzpicture}
	\caption{When $(i,d)$ is dominated by $(i,d-1)$,
          diagonal $d-1$ can have a match at row $i+1$ (shown on left)
          or a mismatch (on right). 
          Edges used by a shortest path to $(i+1,d)$ are drawn as solid, 
          and other edges are dashed.
        }
	\label{fig:domination}
	\vspace{0.1cm}
        \hrule
	\vspace{-0.1cm}
\ifprocs
\end{figure*}
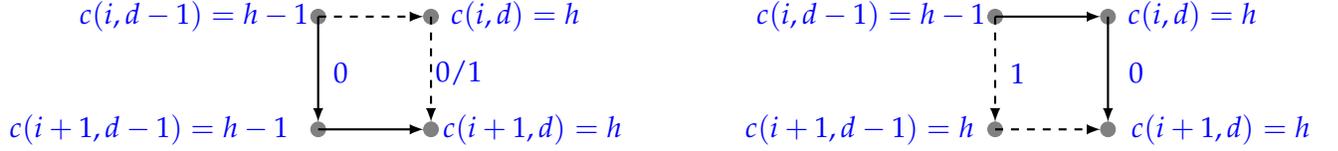
\else
\end{figure}
\fi

\paragraph{Potent vertices}
To formalize the above observation, we now define potent vertices
and assert that it suffices to inspect only such vertices.

\begin{definition} %
	\label{def:potent} 
	We say that diagonal $d$ is \emph{potent} at row $i$ 
	if it satisfies the two requirements: 
	\begin{itemize} \compactify
		\item if $(i,d)$ is dominated by $(i,d-1)$,
		then we require that diagonal $d-1$ is potent at row $i$
		and has a mismatch at row $i+1$; and
		\item if $(i,d)$ is dominated by $(i-1,d+1)$
		then we require that diagonal $d+1$ is potent at row $i-1$
		and has a mismatch at row $i$.
	\end{itemize}
\end{definition}

Notice that if $(i,d)$ is not dominated, 
then both requirements are vacuous, and thus $(i,d)$ is potent.
In particular, the source $(0,0)$ is not dominated and thus potent. 

The three lemmas below show that information whether vertices are potent 
is very useful for computing shortest paths from the source,
i.e., vertex costs.
Informally, the first lemma shows that our algorithm can restrict its attention
to paths that consist of potent vertices 
(at least one, because the source is potent by definition) 
followed by non-potent vertices (zero or more).
In particular, for a potent vertex,
the path to it passes only through potent vertices. 
The next two lemmas show that information whether a vertex is potent or not 
can sometimes make the computation of vertex costs trivial. 
The proofs of all three lemmas appear in Appendix~\ref{app:appendixPotential}.

\begin{lemma}
\label{lem:potential}
Every vertex in the grid graph $G_{x,y}$ has a shortest path from $(0,0)$, 
in which non-potent vertices appear only after potent vertices. 
\end{lemma}

\begin{lemma}
\label{lem:imp}
If $(i,d)$ is non-potent, then $c(i+1,d)=c(i,d)$.
\end{lemma}

\begin{lemma}
\label{lem:mismatchPotent}
Suppose $(i,d)$ is potent.
If $(i+1,d)$ has a mismatch, then $c(i+1,d)=c(i,d)+1$.
Otherwise (it has a match), $c(i+1,d)=c(i,d)$ and $(i+1,d)$ is potent.
\end{lemma}

\begin{remark}
Lemmas~\ref{lem:potential}, \ref{lem:imp} and~\ref{lem:mismatchPotent}
hold also for a generalized grid graph
(as defined in Section~\ref{sec:GenGridGraph}).
\end{remark}

The next challenge is to find the potent vertices algorithmically
without scanning the entire grid graph.

\paragraph{Active vertices} 
We now describe an algorithm that computes iteratively for each row $i=0,1,\ldots,n$ 
a list $\Diag_i$ of diagonals that we call \emph{active} diagonals. 
The idea is that $\Diag_i$ will be a superset of the potent diagonals
at row $i$, 
and that scanning each list $\Diag_i$ 
will create $\Diag_{i+1}$ for the next row. 
However, this description is an over-simplification, 
because $\Diag_i$ itself might change while being scanned, as explained next. 

Formally, the algorithm starts with the following initialization.
It sets $\Diag_0=\set{0}$ and every other list $\Diag_i=\emptyset$.
It then creates an array $\cA$ to store the costs of every diagonal
(at the current row, see more below),
and sets $\cA[d]=\abs{d}$ for every $d\in[-n..n]$. 

The algorithm then iterates over the rows $i=0,1,\ldots,n$,
where each iteration $i$ scans $\Diag_i$ in increasing order. 
To process each scanned $d\in\Diag_i$, the algorithm first determines
whether $(i,d)$ is potent using Definition~\ref{def:potent},
except for vertex $(0,0)$ which is potent by definition.
This requires comparing the cost of $(i,d)$ to its two in-neighbors. 
\footnote{If an in-neighbor does not exist (e.g., out of range),
  we consider it to be  non-dominating.}
We show in Lemma~\ref{lem:costCorrectness} that at this point we have: 
$\cA[d]=c(i,d)$, $\cA[d-1]=c(i+1,d-1)$, and $\cA[d+1]=c(i,d+1)$. In order to check
whether $(i,d)$ is potent, we need the values of $c(i,d-1)$ and $c(i-1,d+1)$ along with the information
whether $(i, d-1)$ and $(i-1,d+1)$ are potent. By maintaining $\cA[d]$ values for the last two computed rows and two simple boolean arrays to indicate whether a diagonal $d'$ is potent at the current row and the previous row (by default a diagonal is not potent), the necessary information to compute whether $(i,d)$ is dominated and potent is available to us.

If $(i,d)$ is dominated, the algorithm further checks
whether each dominating in-neighbor is potent
and whether it has a mismatch at the relevant row.%
Now $\cA$ is updated: 
if $(i,d)$ is determined to be potent and $(i+1,d)$ has a mismatch,
then $\cA[d]$ is incremented by $1$.

If $(i,d)$ is not potent, then $d$ is discarded from $\Diag_i$, 
and the algorithm proceeds to scan the next diagonal in $\Diag_i$.
Otherwise, i.e., $(i,d)$ is potent, 
the algorithm adds $d$ to $\Diag_{i+1}$
(because it may be potent at row $i+1$),
and then checks whether $(i+1,d)$ has a mismatch; 
if it has, then $d+1$ is added to $\Diag_i$ and $d-1$ added to $\Diag_{i+1}$
(because these vertices may be potent).
Observe that processing $d\in\Diag_i$ may cause adding to $\Diag_i$
diagonal $d+1$, which obviously should be processed next; 
this means that $\Diag_i$ is scanned adaptively, but still in increasing order. 

The correctness of this procedure is based on the next two lemmas,
whose proofs appear in Appendix~\ref{app:appendixPotential}.

\begin{lemma}\label{lem:ComputingPotential}
In this procedure, every potent vertex $(i,d)$
is eventually inserted to $\Diag_i$.
Moreover, at the end of iteration $i$ (scanning $\Diag_i$), 
the diagonals in $\Diag_i$ represent exactly the potent vertices at row $i$.
\end{lemma}

\begin{lemma}\label{lem:costCorrectness}
Prior to iteration $i+1$ (or equivalently after iteration $i$),
every $\cA[d']$ stores the value $c(i+1,d')$. 
More precisely, after processing diagonal $d\in\Diag_i$, 
each $\cA[d']$ has value 
$c(i,d')$ for $d'>d$ and 
$c(i+1,d')$ for $d'\le d$.
\end{lemma}

\subsection{Transitioning between Sampling Modes}
\label{sec:modetransition}

Recall that the algorithm presented in Section~\ref{sec:GridSampling} has asymmetric query complexity --- it samples string $x$ uniformly at rate $\frac{\log n}{t}$ but may query the entire string $y$, as it compares each sampled coordinate $x_i$ against $y_{i+d}$ for all possible $d\in \pmt$.
In order to improve the query complexity to $\tilde{O}(\frac{n}{t}+t^3)$ and prove Theorem~\ref{thm:main}, our algorithm will alternate between uniform sampling and contiguous (non-uniform) sampling.
However, during the sampling mode, it is still prohibitive to compare
each sampled coordinate $x_i$ against $y_{i+d}$ for all possible $d\in \pmt$.
Instead, we would like to leverage the information given by $\Diag_i$. 
If $\card{\Diag_i}=1$, it suffices to compare $x_i$ against $y_{i+d}$
only for the single $d\in \Diag_i$.
And if $\card{\Diag_i} > 1$,
then Lemma~\ref{lemma:periodicity} (part a) below guarantees that 
both $x$ and $y$ follow the same periodic pattern (coming into the current row),
in which case we switch to uniform sampling,
and merely ``verify'' that the periodicity continues (in the next rows)
by comparing these samples to our pattern. 
Obviously, this verification is probabilistic, 
but by Claim~\ref{claim:badSamplingEvent}, with high probability
it holds up to $O(t)$  Hamming errors (equivalently, character substitutions). 
When we see a sample that deviates from the periodicity in either $x$ or $y$,
we recompute $\Diag_i$ and start over,
using Lemma~\ref{lemma:periodicity} (part b)
to argue that almost all diagonals in $\Diag_i$ ``must see'' edit operations,
which ``count against'' our budget of $t$ edit operations.

\begin{lemma}\label{lemma:periodicity}
	Let $x,y\in \Sigma^n$, let $i\in [n]$ and let $\Diag\subseteq \pmt$
	be a set of diagonals of size $\card{D}>1$.
	Define $g=\gcd\set{d-d':\ d>d'\in \Diag}$, $p=x_{[i-g+1..i]}$
	and $m=\max\Diag - \min\Diag$. 
	
	\begin{enumerate}
        \item[(a)] If 
		\begin{equation} \label{eq:diag}
		\forall d\in\Diag, \qquad
		x_{[i-2m+1.. i]} = y_{[i-2m+1.. i]+d} 
		\end{equation}
		then $x_{[i-2m+1..i]}$
		and $y_{[i-2m+1+\min{\Diag}..i+\max{\Diag}]}$ 
		are both periodic with the same period pattern $p$. %
		
       \item[(b)] Assume the conclusion of part (a) holds
                (i.e., $x_{[i-2m+1.. i]}$ and $y_{[i-2m+1+\min{\Diag}.. i+\max{\Diag}]}$ are both periodic with same period pattern $p$)
		but either $x_{i+1}\neq p_{1}$ or $y_{i+1+\max \Diag}\neq p_{1}$ (observe that $p_1 = p_{(i+1-2m+1)\bmod g}$).
		Then each diagonal in $\Diag$ except perhaps at most one, 
		has a mismatch in at least one of the $m$ rows $i+1,\ldots,i+m$.
	\end{enumerate}

\end{lemma}

The proof of this lemma appears in Section~\ref{sec:analysis}, 
after the description of our algorithm.
We remark that part (a) has been established in \cite{CGK16b}
and was used there for a different purpose.

\subsection{Our Sublinear Algorithm}\label{sec:AlgorithmDescription}

At a high level, the algorithm first picks a random set $S\subseteq[n]$,
by including in $S$ each row independently with probability $\frac{\log n}t$, 
and then proceeds in rounds, where each round processes one row. 
When a round processes row $i\in[n]$,
we will say that it scans row $i$, 
as it will process some (but not all) of its vertices $(i,d)$,
always in increasing order of $d$.
We shall also call it round $i$ (although it need not be the $i$-th round).

The algorithm has two modes of scanning the rows, 
called \emph{contiguous} and \emph{sampling}. 
Roughly speaking, the sampling mode scans only sampled rows, i.e., $i\in S$,
and at each such round it reads only two characters $x_i$ and $y_{i+d}$ 
for a \emph{single} active diagonal $d\in \Diag_i$.
These two characters are compared not to each other (unless $|\Diag_i|=1$),
but rather to a pattern determined in previous rounds.
The goal is to examine whether a certain periodicity in $x$ and $y$ is broken,
in which case the algorithm switches to the contiguous  mode. 
Observe that the sampling mode is rather lightweight -- 
it scans rows at a rate of $1/t$ and performs minimal computation per row. 

In contrast, the contiguous mode scans the rows one by one,
and at each such round $i$, it compares $x_i$ to $y_{i+d}$
for every active diagonal $d\in\Diag_i$.

This mode is applied in bulks of at least $O(t)$ consecutive rows, 
until in a bulk of $O(t)$ the set of active diagonals does not change 
(which means that none of the active diagonals sees a mismatch along these lines).
In the case that the set of active diagonals does not change we deduce a periodicity
structure on the corresponding parts of $x$ and $y$ and the algorithm switches to the sampling mode. 

We can now describe the algorithm in full detail. 
Let $i_1<\ldots <i_{\abs{S}}$ denote the rows selected into $S$. 
While scanning the rows, the algorithm maintains a counter, 
$i$ for the current row 
initialized to $i=i_1$. 
as the algorithm is started in the sampling mode.
The algorithm maintains also two lists of diagonals, 
$\Diag_i$ of active diagonals at the current round $i$
and $\Dnext$ that is constructed for the next round 
(which is either $i+1$ or the next row in $S$, depending on the mode),
initialized to $\Diag_{i_1}=\set{0}$ and $\Dnext=\emptyset$. 
During its execution, the algorithm computes and stores
an array $\cA$ to store an estimate of cost of every diagonal 
at the current row. It is initialized by setting  $\cA[d]=\abs{d}$ for every $d\in[-t..t]$.

\paragraph{The Contiguous Mode}
In the contiguous mode, the algorithm scans the rows one by one. 
At each round $i$ in this mode,
the algorithm scans the active diagonals $d\in\Diag_i$ in increasing order,
and each such $d$ is processed as follows. 
If $\cA(d)>t-\abs{d}$ then this diagonal $d$ is ignored,
i.e., its processing is concluded. 
Otherwise, the algorithm checks whether diagonal $d$ is potent for row $i$,
as defined in Definition~\ref{def:potent}
but with respect to cost function $\cA$ (instead of $c$),
implemented by comparing $\cA[d]$ to the value of $\cA[d+1]$ and $\cA[d-1]$
in previous rows.%
\footnote{This check is implemented similarly to Section~\ref{sec:LMS},
  by maintaining the costs computed in the previous two rows and the list of potent diagonals at those rows.
}
If $d$ is indeed potent, then it is added to $\Dnext$,
and if in addition, $x_{i+1}\neq y_{i+d+1}$,
then $d+1$ is added to $\Diag_i$ and $d-1$ 
is added to $\Dnext$ and $\cA[d]$ is incremented by $1$.
This completes the processing of $d\in\Diag_i$.

When the algorithm finishes scanning $\Diag_i$,
it has to decide about the next round. 
If along the last $2(\max \Diag_i-\min \Diag_i)$ rows there exists a row $i'$,
in which a mismatch was found at some diagonal $d \in \Diag_{i'}$, 
then the algorithm increments $i$ by $1$, sets $\Diag_i=\Dnext$ and $\Dnext=\emptyset$, 
and proceeds to the next round (for this new value of $i$),
staying in the contiguous mode. 

Otherwise (no mismatch was found along these rows),
the algorithm prepares to switch to the sampling mode
by computing three variables:
\begin{align*}
  &g = \gcd \set{d-d'}_{d\neq d' \in \Diag_i} 
  \\
  &\ipat = i-2(\max \Diag_i-\min \Diag_i)+1 , \\
  &p = x_{[\ipat..\ipat+g-1]}, 
\end{align*}
and then increases $i$ to the next row in $S$ (smallest one after the current $i$),
sets $\Diag_i=\Dnext$ and $\Dnext=\emptyset$,
and proceeds to the next round (for this new value of $i$) but in the sampling mode.

\paragraph{The Sampling Mode}
In the sampling mode, the algorithm processes only sampled rows $i\in S$.
The sampling mode performs one of two checks,
either \emph{periodicity check} or \emph{shift check},
depending on whether $\card{\Diag_{i}}>1$ or not.

\textbf{(i) Periodicity check:}
This check is applied only if $|\Diag_i| > 1$. 
The algorithm first checks whether
both $x_{i+1}$ and $y_{i+\max \Diag_i+1}$ are equal to $p_{(i-\ipat+1)\bmod g}$.
If both match (are equal),
the algorithm finds the next row $\inext$ in $S$ (smallest one after the current $i$), 
sets $\Diag_{\inext}=\Diag_{i}$, increases $i$ to $\inext$, 
and proceeds to the next round (for this new value of $i$),
still in the sampling mode (to perform a periodicity check
because again $\card{\Diag_{i}}>1$).

Otherwise (at least one of the two comparisons fails),
the algorithm employs binary search to detect
a row $j\in [\ipat+2(\max \Diag_i -\min \Diag_i).. i] $
with a ``period transition'', defined as $j$ satisfying the two conditions:
\begin{enumerate}
\item 
  for all $j' \in [j-2(\max \Diag_i-\min \Diag_i)..j]$,
  we have 
  $x_{j'}=y_{j'+\max \Diag_i}=p_{(j'-\ipat) \bmod g }$; and 
\item
  either $x_{j+1}$ or $y_{j+\max \Diag_i +1}$
  is not equal to $p_{(j+1-\ipat) \bmod g}$. 
\end{enumerate}
We later prove that such a $j$ must exist. 
The algorithm then finds all the diagonals $d\in \Diag_i$ that have
a mismatch in at least one row in the range $[j..j+\max \Diag_i-\min \Diag_i]$.
Lemma~\ref{lemma:periodicity} shows that this event (at least one mismatch)
must occur for all the diagonals in $\Diag_i$ except perhaps one diagonal,
which we denote by $d^*$ (if exists).
Now the algorithm sets
\begin{equation}
  \forall d\in \Diag_i, d\neq d^*,
  \qquad 
  \cA[d]=\cA[d]+1; 
\end{equation}  
and adds $d$, $d+1$, and $d-1$ to $\Dnext$.

If $d^*$ exists, the algorithm further samples 
a new set $S^* \subseteq [\ipat..i]$ at rate $\frac {\log n}{t}$,
and for each row $j'\in S^*$ it compares $x_{j'}$ to $y_{j'+d^*}$. 
If no mismatch is found in $S^*$, 
then the algorithm 
adds $d^*$  to $\Dnext$. 
Otherwise (a mismatch is found), it sets $\cA[d^*]=\cA[d^*]+1$
and adds $d^*$, $d^*+1$, and $d^*-1$ to $\Dnext$.

Finally, the algorithm 
increments $i$ by $1$, 
sets $\Diag_i=\Dnext$ and $\Dnext=\emptyset$,
and proceeds to the next round (for this new value of $i$)
but in the contiguous mode.

\textbf{(ii) Shift check:}
This check is applied only if $|\Diag_i|=1$;
let $d$ be the unique diagonal in $\Diag_{i}$. 
The algorithm compares $x_{i+1}$ and $y_{i+1+d}$. 
If they match (are equal),
the algorithm increases $i$ to the next row in $S$ (smallest one after the current $i$), 
sets $\Diag_i=\set{d}$ 
and proceeds to the next round (for this new value of $i$),
still in the sampling mode (to perform a shift check
because again $\card{\Diag_i}=1$).

Otherwise (they do not match), the algorithm sets
$\cA[d]=\cA[d]+1$
and adds $d$, $d+1$, and $d-1$ to $\Dnext$.
The algorithm then increments $i$ by $1$,
sets $\Diag_i=\Dnext$ and $\Dnext=\emptyset$,
and proceeds to the next round (for this new value of $i$)
but in the contiguous mode.

\paragraph{Stopping Condition}
The algorithm halts and outputs $\far$ if at any point the value of $\cA[0]$ reaches $t+1$.
If the algorithm completes processing the rows 
(the last one is in $S$ or in $[n]$, depending on the mode),
and still $\cA[0]\le t$,
then the algorithm halts and outputs $\close$.

\subsection{Analysis} \label{sec:analysis}
Let us start by proving lemma~\ref{lemma:periodicity}. For convenience, let us restate it:

\begin{lemma*}[\ref{lemma:periodicity}]
	Let $x,y\in \Sigma^n$, let $i\in [n]$ and let $\Diag\subseteq \pmt$
	be a set of diagonals of size $\card{D}>1$.
	Define $g=\gcd\set{d-d':\ d>d'\in \Diag}$, $p=x_{[i-g+1..i]}$
	and $m=\max\Diag - \min\Diag$. 
	
	\begin{enumerate}
        \item[(a)] If 
		\begin{equation} \label{eq:diag1}
		\forall d\in\Diag, \qquad
		x_{[i-2m+1.. i]} = y_{[i-2m+1.. i]+d} 
		\end{equation}
		then $x_{[i-2m+1..i]}$
		and $y_{[i-2m+1+\min{\Diag}..i+\max{\Diag}]}$ 
		are both periodic with the same period pattern $p$. %
		
       \item[(b)] Assume the conclusion of part (a) holds
                (i.e., $x_{[i-2m+1.. i]}$ and $y_{[i-2m+1+\min{\Diag}.. i+\max{\Diag}]}$ are both periodic with same period pattern $p$)
		but either $x_{i+1}\neq p_{1}$ or $y_{i+1+\max \Diag}\neq p_{1}$ (observe that $p_1 = p_{(i+1-2m+1)\bmod g}$).
		Then each diagonal in $\Diag$ except perhaps at most one, 
		has a mismatch in at least one of the $m$ rows $i+1,\ldots,i+m$.
	\end{enumerate}

\end{lemma*}

\begin{proof}
  To prove the first part, consider two diagonals $d_1<d_2$ in $\Diag$.
  We can see that $x_{[i-2m+1..i]}$ has period length $d_2-d_1$, 
  by using~\eqref{eq:diag1} twice (for all relevant positions $j$)
  \begin{equation*} 
    \forall j\in [i-2m+1..i-(d_2-d_1)],
    \qquad
    x_j = y_{j+d_2} = x_{j+(d_2-d_1)} .
  \end{equation*}
  Now use the fact that if a string $s$ is periodic
  with two period lengths $l\neq l'$, 
  then it is also periodic with period length $\gcd\set{l,l'}$.%
  \footnote{To see this,
    use Bézout's identity to write $\gcd\set{l,l'}=tl+t'l'$ for integers $t,t'$, 
    then show that $s_j = s_{j+tl+t'l'}$
    by a sequence of $|t|+|t'|$ equalities, 
    ordered so as to stay inside $s$. 
  }
  It follows that $x$ has period length $g=\gcd\set{d-d':\ d>d'\in \Diag}$,
  hence its period pattern is $p=x_{\set{i-g+1,\ldots,i}}$. 

  By applying a similar argument to $y$,
  we obtain that it too has period length $g$,
  and its period pattern is $y_{[i+\max\Diag-g+1..i+\max\Diag]}$. 
  Moreover, by applying~\eqref{eq:diag1} to $d=\max\Diag$
  we see that the period patterns of $x$ and of $y$ are equal 
  \[
    p
    = x_{[i-g+1..i]}
    =
    y_{[i+\max\Diag-g+1..i+\max\Diag]}. 
  \]

\ifprocs
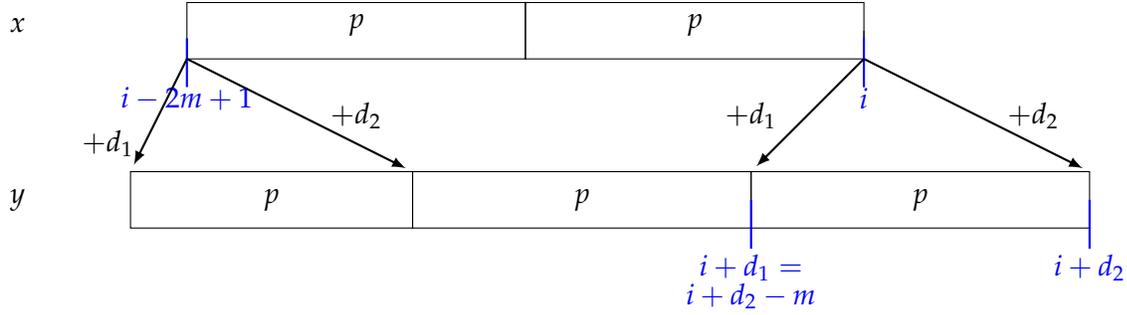
\begin{figure*}
\else
\begin{figure}[ht]
\fi
			\begin{tikzpicture}[scale=1.5,shorten >=1mm,>=latex]
			\tikzstyle gridlines=[color=black!20,very thin]
			\node at (-2,0.3) {$x$};
			\node at (-2,-1.25) {$y$};

			\draw[] (-0.5,0) rectangle (2.5,0.5);
			\draw[] (2.5,0) rectangle (5.5,0.5);
			\draw[] (5.5,-0.25) rectangle (5.5,0.25);

			\node at (1,0.3) {$p$};
			\node at (4,0.3) {$p$};
			\draw[-,thick,blue] (5.5,-0.25)--(5.5,0.25);
			\node[blue] at (5.5,-0.35) {$i$};
			
			\draw[-,thick,blue] (-0.5,-0.25)--(-0.5,0.25);
			\node[blue] at (-0.5,-0.35) {$i-2m+1$};

			\draw[] (-1,-1.5) rectangle (1.5,-1);
			\draw[] (1.5,-1.5) rectangle (4.5,-1);
			\draw[] (4.5,-1.5) rectangle (7.5,-1);
			
			\node at (0.25,-1.25) {$p$};
			\node at (3,-1.25) {$p$};
			\node at (6,-1.25) {$p$};
			
			\draw[-,thick,blue] (7.5,-1.25)--(7.5,-1.75);
			\node[blue] at (7.5,-1.85) {$i+d_2$};
			
			\draw[-,thick,blue] (4.5,-1.25)--(4.5,-1.75);
			\node[blue] at (4.5,-1.85) {$i+d_1=$};
			\node[blue] at (4.5,-2.1) {$i+d_2-m$};

			\draw[->,thick] (5.5,0)--(7.5,-1);
			\node at (7,-0.5) {$+d_2$};
			\draw[->,thick] (5.5,0)--(4.5,-1);
			\node at (4.5,-0.5) {$+d_1$};

			\draw[->,thick] (-0.5,0)--(-1,-1);
			\node at (1,-0.5) {$+d_2$};
			\draw[->,thick] (-0.5,0)--(1.5,-1);
			\node at (-1.2,-0.75) {$+d_1$};
			\end{tikzpicture}

			\caption{Matching the strings along different diagonals $d_1,d_2$.}
			\label{fig:periodicity}
\ifprocs
\end{figure*}
\else
\end{figure}
\fi

Let us now prove the second part. 
Assume first that $x_{i+1}\neq p_{1}$.
Then for every diagonal $d\in \Diag\setminus \set{\max\Diag}$
we have $y_{i+d+1} = y_{i+\min\Diag+1} = p_1$
(because $d-\min\Diag$ is a multiple of the period $g$
and all these positions are inside the periodic part of $y$),
and we see that diagonal $d$ has a mismatch at row $i+1$.

Next, assume that $x_{i+1}= p_{1} \neq y_{i+\max \Diag+1}$. 
Let $i_x\ge i+1$ be the smallest row ``deviating'' from the pattern $p$,
i.e., $x_{i_x}\neq p_{(i_x-i) \bmod g}$,
letting $i_x=\infty$ if no such row exists.
Our assumption implies that actually $i_x > i+1$. 
We proceed by a case analysis.

If $i_x \ge i+1+(\max \Diag-\min \Diag)$,
then we can show that every diagonal $d\in \Diag\setminus\set{\min\Diag}$
has a mismatch at row $i+1+(\max \Diag-d)$.
Indeed, $x$ is periodic up to that row
because $i+1+(\max \Diag-d) < i+1+(\max\Diag - \min\Diag) \le i_x$,
and thus
\[
  x_{i+1+\max\Diag-d}= x_{i+1}
  = p_{1}
  \neq y_{i+1+\max \Diag}.
\]
And since $0\le \max \Diag-d < m$,
the row with the mismatch is indeed in the claimed range $i+1,\ldots,i+m$.

Otherwise, we have $i+1<i_x<i+\max \Diag-\min \Diag+1$.
For every $d\in \Diag$ satisfying $d > d':=i+1+ \max\Diag-i_x$ 
there is a mismatch at row $i+1+\max \Diag-d$ 
as $x$ is still periodic at that position because 
$i+1+\max\Diag-d < i+1+\max \Diag-d' = i_x$, 
and thus
\[
  x_{i+1+\max\Diag-d} = x_{i+1}
  = p_1
  \neq y_{i+1+\max\Diag}. 
\]
For every $d\in \Diag$ satisfying $d<d'$,
we have a mismatch at row $i_x$,
because $x$ is not periodic at $i_x$
and thus $x_{i_x} \neq p_{(i_x-i)\bmod g}$, 
while $y$ is still periodic at position 
$i_x+d < i_x+d' = i+1+ \max\Diag$ and thus 
\[
  y_{i_x+d}
  = p_{(i_x+d - (i+\min\Diag)) \bmod g}
  = p_{(i_x-i)\bmod g}
  \neq x_{i_x}. 
\]
In both cases, $d>d'$ and $d<d'$, 
the mismatched row is indeed in the claimed range $i+1,\ldots,i+m$, 
and together the two cases include all but at most one diagonal in $\Diag$.

\ifprocs
\begin{figure*}
\else
\begin{figure}[ht]
\fi

		\begin{tikzpicture}[scale=1.5,shorten >=1mm,>=latex]
		\tikzstyle gridlines=[color=black!20,very thin]
		\node at (-0.5,0.3) {$x$};
		\node at (0.5,-1.25) {$y$};		
		
		\draw[] (0,0) rectangle (3,0.5);
		\node at (1.5,0.3) {$p$};

	\node at (4.5,-1.25) {$p$};

		\node at (0,0.7) {$i+1$};
		\node at (3,0.7) {$i_x$};
		\draw[] (3,0) rectangle (6,0.5);

		\draw[] (3,-1.5) rectangle (6,-1);
		\draw[] (6,-1.5) rectangle (9,-1);
		\node at (6,-0.7) {$i+\max \Diag+1$};

		\draw[->,dashed] (3,0)--(6,-1);
		\node at (4.5,-0.4) {$d'$};
		
		\draw[->,thick,color=red] (2,0)--(6,-1);
		\draw[->,thick,color=red] (1,0)--(6,-1);
		
		\draw[->,thick,color=blue] (3,0)--(5,-1);
				\draw[->,thick,color=blue] (3,0)--(4,-1);

		\end{tikzpicture}
		\caption{Red lines represent mismatches on diagonals $d>d'$, and blue lines represent mismatches on diagonals $d<d'$}
		\label{fig:periodicityViolation}

\ifprocs
\end{figure*}
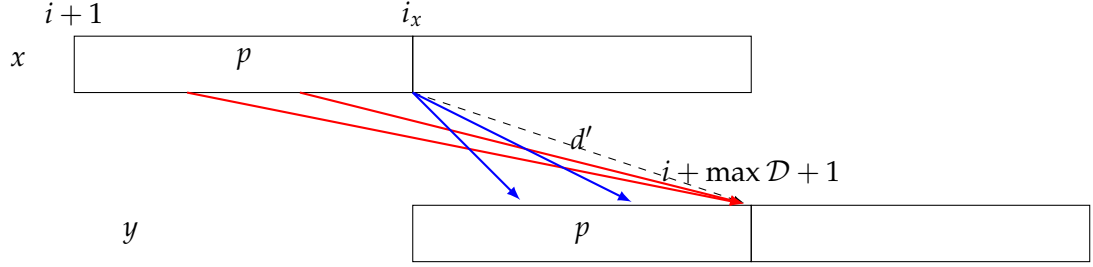
\else
\end{figure}
\fi
              
\end{proof}

We show next that when the algorithm is in the sampling mode and executes a binary search to find a period transition, then it always succeeds.

\begin{claim}\label{claim:periodViolation}
	Suppose at round $i$ the algorithm stays at the sampling mode performing a periodicity check and it detects a period violation, that is either $x_{i+1}\neq p_{(i+1-\ipat)\bmod g }$ or $y_{i+\max \Diag+1} \neq p_{(i+1-\ipat)\bmod g}$, then:
	
	\begin{enumerate}
		\item 
                  There exists a row $j\in [i_{pat}+2(\max \Diag -\min \Diag)..i] $ with a period transition,
                  that is, 
\ifprocs
                  \begin{align*}
                  \forall j' \in & [j-2(\max \Diag-\min \Diag)..j],\\
                  x_{j'} & =y_{j'+\max \Diag} =p_{(j'-i_{pat})\bmod g } \text{ and } \\
                  &\text{ either } x_{j+1}\neq p_{(j+1-i_{pat})\bmod g} \\
                  &\text{ or } y_{j+\max \Diag +1}\neq p_{(j+1-i_{pat})\bmod g}
                \end{align*}
\else
                  \begin{align*}
                  \forall j' \in & [j-2(\max \Diag-\min \Diag)..j], \quad
                  x_{j'}  =y_{j'+\max \Diag} =p_{(j'-i_{pat})\bmod g } \text{ and } \\
                  &\text{ either } x_{j+1}\neq p_{(j+1-i_{pat})\bmod g} 
                  \text{ or } y_{j+\max \Diag +1}\neq p_{(j+1-i_{pat})\bmod g}
                \end{align*}
\fi

		\item Let $j$ be the row from the previous item, then for all $d\in \Diag$ but perhaps at most one diagonal $d^*$, there exists a row  $j'\in [j..j+\max \Diag -\min \Diag]$ such that  $x_{j'}\neq y_{j'+d}$.
	\end{enumerate}
\end{claim}

\begin{proof}
We first prove part (1).
Since the algorithm enters the sampling mode only after on the last $2(\max \Diag-\min \Diag)$ (where $\Diag$ denotes the set of potent diagonals when the algorithm enters the sampling mode) no mismatch has occurred for each $d$ in the potent diagonals, then our set $\Diag$ satisfies the condition of Lemma~\ref{lemma:periodicity} at row $\ipat$. Note that by
          Lemma~\ref{lem:mismatchPotent}, %
          the set $\Diag$ did not change. By the lemma we get $x_{[\ipat..\ipat+2(\max \Diag-\min \Diag)]}=p\circ \cdots \circ p$ and $y_{[\ipat..\ipat+2(\max \Diag-\min \Diag)]+\max \Diag}=p\circ \cdots \circ p$. Since at row $i+1$ at least one of the values $x_{i+1}, y_{i+\max \Diag+1}$ does not match the corresponding character in $p$, then between $i$ and $\ipat$, there exists an index $j$ on which we have a period transition.

Part (2) follows directly by the second item of Lemma~\ref{lemma:periodicity}.
\end{proof}

\paragraph{Correctness Analysis}
We now analyze the correctness of the algorithm in the two cases, 
$\close$ that is $\ed(x,y)\le t/2$ and $\far$ that is $\ed(x,y)> 13t^2$, more specifically we prove the following lemmas:
\begin{lemma}\label{claim:close}
	Let $x,y\in \zo^n$. If $\ed(x,y)\le t/2$ then with probability $1$ the algorithm outputs $\close$.	
\end{lemma}

\begin{lemma}
	\label{claim:far}
	If $\ed(x,y)> 13t^2$ then with probability at least $2/3$ the algorithm outputs $\far$.	
\end{lemma}

\subsection{Proof of Close Case $\ed(x,y)\le t/2$}

\begin{proof}[Proof of Lemma~\ref{claim:close}]
	Let $x,y$ be such that $\ed(x,y)\le t/2$. We show that with probability $1$ the algorithm outputs $\close$.
	For this sake, we build a grid graph $G'=(V',E')$ and define a new cost function $c':E'\to \N \cup \set{0}$. 
	We show that there exists a path connecting the sink and the source  in $G'$ whose cost $c'$ is at most $t/2$.
	The set $V'$ will include all the vertices that the algorithm scans. 
	Then, we define a new cost function, denoted $\cALG$ on the $E'$. 
	We claim that the costs assigned by the algorithm are consistent with $\cALG$. 
	Finally, we connect the costs $c'$ and $\cALG$.
	
	\paragraph*{Graph Construction.} The graph is built as follows: Recall $S$ denote the sampled rows by the algorithm. 
	Let $S'$ be the set of rows $i$ such that either $i\in S$ or the algorithm scans the row $i$ during the contiguous mode or $i$ is a row that the algorithm scans after finding a period transition during binary search or during a shift check.
	We set $$V'=S'\times \pmt \cup \set {(0,0), (n,0)}.$$
	
	Let us describe the cost $c'$ for each edge in $G'$.
	Let $i\in S'$. Let us first deal with the boundaries. We connect $(0,0)$ to $(i_1,0)$ where $i_1$ is the smallest row in $S'$ by an edge of cost $0$. Let $i_{\abs{S'}}$ be the largest row in $S'$, we connect each vertex  $(i_{\abs{S'}},d)$, into $(n,0)$ by an edge of cost $\abs{d}$.
	
	Let $i\in S'$ which is not the last row in $S'$ and let $\inext$ be the next row in $S'$. Each vertex $(i,d)$ is connected to:
\ifprocs
\begin{align*}
  (i,d)&\to (i,d+1), \\
  (i,d)&\to (\inext,d-1) \text{ and } \\
  (i,d)&\to (\inext,d).
\end{align*}
\else
\begin{align*}
  (i,d)\to (i,d+1),
  (i,d)\to (\inext,d-1) \text{ and }
  (i,d)\to (\inext,d).
\end{align*}
\fi
        The first two edges are associated with cost $1$.
	The cost of $(i,d)\to (\inext,d)$ is defined as: $\indicator{\set{x_{i+1}\neq y_{i+d+1}}}$. 
	
\begin{corollary}\label{cor:c'}
	For every path $\tau$ in $G_{x,y}$ from $(0,0)$ to $(n,0)$, there exists a corresponding path $\tau_S$ in $G'$ from $(0,0)$ to $(n,0)$ such that $c'(\tau_S) \leq c(\tau)$.
\end{corollary}	
This follows directly by Claim~\ref{claim:closeCaseAnalysis}. 
Next we would like to connect the costs $\cA$ and $c'$. 
For this sake we first define a cost function $\cALG$  on the edge set $E'$ 
and then claim it is consistent with the costs assigned by the algorithm.

\paragraph{Defining $\cALG$} 

	We define a cost function $\cALG$ on the edges of $G'$ as follows.

	Let $i\in S'$. Let us first deal with the boundaries. We connect $(0,0)$ to $(i_1,0)$ where $i_1$ is the smallest row in $S'$ by an edge of cost $0$. 
	Let $i_{\abs{S'}}$ be the largest row in $S'$, we connect each vertex  $(i_{\abs{S'}},d)$, into $(n,0)$ by an edge of cost $\abs{d}$.
	
	Let $i\in S'$ which is not the last row in $S'$ and let $\inext$ be the next row in $S'$.
	The edges $(i,d)\to (i,d+1), (\inext,d+1)\to (i,d) $ are associated with cost $1$.
	The cost of $(i,d)\to (\inext,d)$ is defined as follows:
	
	\textbf {Case 1: $i$ is a contiguous round:} Then its cost is $\indicator{x_{i+1}\neq y_{i+d+1}}$. 
	
	\textbf {Case 2: $i$ is a sampling round:} If the consistency check passes then the cost is $0$.
	Otherwise, recall that the algorithm first detects a row  
	$j$ on which there is a period transition and then finds all the diagonals $d\in \Diag$ that have
	a mismatch in at least one row in the range $[j..j+\max \Diag-\min \Diag]$. For each of these diagonals we assign $\cALG((i,d)\to (\inext,d))=1$. 
	Recall that if $d$ does not have a mismatch in either of these rows in this range,  then the algorithm samples another set 
	$S^*$ and checks whether $d$ has a mismatch in $S^*$. 
	If the later test passes then the cost is $0$ and otherwise it is $1$. For the rest of the diagonals ($d\notin \Diag$) we set $\cALG((i,d)\to (\inext,d))=0$. 
	
	We next connect $\cA$ and $\cALG$, for this we extend $\cALG$ to $V'$ by setting $\cALG(v)$ 
	as the shortest path cost (with respect to $\cALG$) connecting $(0,0)$ and $(i,d)$. We next prove:
	\begin{lemma}\label{lem:cALG}

		Let $i$ be a row scanned by the algorithm, prior to iteration $\inext$ (or equivalently after iteration $i$),
		every $\cA[d']$ stores the value $\cALG(\inext,d')$. 
		More precisely, after processing diagonal $d\in\Diag_i$, 
		each $\cA[d']$ has value 
		$\cALG(i,d')$ for $d'>d$ and 
		$\cALG(\inext,d')$ for $d'\le d$.
		
	\end{lemma}
	\begin{proof}
	To prove the lemma we use the following definition:
	 We say that diagonal $d$ is $\cALG$-\emph{potent} at row $i$ 
	if it satisfies the two requirements: 
	\begin{itemize}
		\item if $(i,d)$ is dominated by $(i,d-1)$ (with respect to $\cALG$),
		then we require that diagonal $d-1$ is potent at row $i$
		and $\cALG((i,d-1)\to (\inext,d-1))=1$; and
		\item if $(i,d)$ is dominated by $(\ilast-1,d+1)$
		then we require that diagonal $d+1$ is potent at row $\ilast$
		and $\cALG((\ilast,d+1)\to (i,d+1))=1$, where $\ilast$ is the largest row in $S'$ smaller than $i$.
	\end{itemize}

We prove only the first assertion,
as the second one is an immediate consequence of it.
The proof is by induction on the grid vertices $(i,d)\in V'$ in lexicographic order
(i.e., their row is the primary key and their diagonal is secondary).
The base case is the time before processing vertex $(0,0)$; 
at this time, $\cA$ stores its initial values, i.e., $\cA[d] = \abs{d}$,
which is equal to $c(0,d)$ for all $d\ge 0$. 
For $d<0$, the base case is the time before processing $(-d,d)$,
because we should only consider vertices reachable from $(0,0)$;
at this time, $\cA[d]=\abs{d}$ is still the initialized value
and it is equal to $c(d,-d)=-d$. 

For the inductive step, we need to show $\cA[d]$ is updated according to $\cALG$.
But using the induction hypothesis, 
we only need to show $\cA[d]$ is updated from $\cALG(i,d)$ to $\cALG(\inext,d)$.
To this end, suppose first that vertex $(i,d)$ is non-potent. 
Then by Lemma~\ref{lem:imp} we have $\cALG(i+1,d)=\cALG(i,d)$, 
(regardless of the cost $\cALG((i,d) \to (\inext,d))$)
and the algorithm indeed does not modify $\cA[d]$. 
Suppose next $(i,d)$ is potent:
Observe that the algorithm increases $\cA[d]$ by $1$ only if
$\cALG((i,d) \to (\inext,d))=1$, in which case by Lemma~\ref{lem:mismatchPotent} we have:  
$\cALG(\inext,d)=\cALG(i,d)+1$, so by the induction hypothesis $\cA[d]$ is incremented to the correct value.
On the other hand, the algorithm does not change the value of $\cA[d]$ by $1$ if
$\cALG((i,d)\to (\inext,d))=0$, in which case by Lemma~\ref{lem:mismatchPotent} we have 
$\cALG(\inext,d)=\cALG(i,d)$, and again by the induction hypothesis $\cA[d]$ matches the correct value.
\end{proof}

	By abusing notations, we define a cost function $c'$ on the vertices of $V'$ by setting $c'(v')$
	as the cost of shortest path connecting $(0,0)$ and $v'$, for each $v'\in V'$.

	Finally we connect the cost $c'$ and $\cALG$, namely: Let $\tau_S$ be a shortest path in $G'$ with respect to the cost function $c'$ connecting $(0,0)$ and $(n,0)$.  We conclude the proof by claiming $\cALG(\tau_S)\le 2 c'(\tau_S)$. 
	\begin{claim}\label{claim:cALGc'}
		$\cALG(\tau_S)\le 2 c'(\tau_S)$.
	\end{claim}
\begin{proof}
	The proof proceeds by induction on the rows in $S'$.
	In particular, we show that for each row $i\in S'$ and each diagonal $d$ traversed by $\tau_S$  at row $i$, we have: $\cALG(i,d)\le 2c'(i,d)$. 
	The base case is $i=0$, the only diagonal $\tau_S$ traverses at row $0$ is $0$ for the cost is $0$ with respect to $\cALG,c'$.
	
	The induction step: let $i\in S'$, let $(\ilast,\dlast)$ be the last vertex on $\tau_S$ before moving to row $i$, and let $(i,d)$ be the first diagonal that $\tau_S$ traverses at row $i$ (observe that $\dlast \in \set{d+1,d}$). We first prove that $\cALG(i,d)\le 2c'(i, d)$. Then we prove it for the rest of the diagonals traversed by $\tau_S$ at row $i$.
	
	\textbf {Case 0- $(\dlast\neq d)$}:
	In this case, $c'(i,d )-c'(\ilast,\dlast) =1$. Since for each neighboring diagonal the difference in $\cALG$ cost is at most $1$, then
	$\cALG(i,d)\leq \cALG(\ilast,\dlast)+1$, the claim follows.
	
	\textbf {Case 1- $i$ is a contiguous round}:

	\textbf {Case 1.2.1-  $(\dlast= d)$ and $(\ilast, \dlast)$ is $\cALG$-potent}:
	In this case $\cALG(i,d)-\cALG(\ilast,\dlast)=c'(i,d)- c'(\ilast,\dlast)=\indicator{x_{\ilast+1}\neq y_{\ilast+\dlast+1}}$, the claim follows.
	
	\textbf {Case 1.2.2-$(\dlast= d)$ and $(\ilast, \dlast)$ is not $\cALG$-potent}:
	In that case from Lemma~\ref{lem:imp}, $\cALG(\ilast,\dlast)=\cALG(i,d)$. 
	In $c'$ the difference between the costs may be either $0$ or $1$, the claim follows.

		\textbf {Case 2- $i$ is a sampling round}: 
	
	\textbf {Case 2.1- $(\dlast= d)$, $i$ is a sampling round and the periodicity check passes at row $i$:}
	Observe that the cost $\cALG$ is not incremented over the last row in either of the diagonals at row $i$. While in $c'$ it may be increased, the claim follows.

	\textbf {Case 2.2- $(\dlast= d)$, $i$ is a sampling round and the periodicity check fails at row $i$:}
	If $(\ilast, \dlast)$ is non-potent then the proof follows by the argument used in case 1.2.2. Let us prove the claim for the case $(\ilast,\dlast)$ being potent. 
	In this case the algorithm performs a binary search to detect a period transition at row $j$. The algorithm sets $\cALG(i,d)=\cALG(\ilast ,d)+1$ if diagonal $d$ has a mismatch at some row $j'\in [j.. j+\max \Diag- \min \Diag]$. For the special diagonal $d^*$ that has no mismatch at none of the rows $j'$, the algorithm samples another set $S'$ and increments the cost if it finds a mismatch along one of the rows in $S'$.
	
	Let us first analyze the diagonals $d$ that have a mismatch at some row $j'$: 
	Recall that $\ipat$ is the first row on which the algorithm switched to sampling mode. 
	First observe that $\ell\in [\ipat.. i)$ we have: $\cALG(\ell,d)=\cALG(\ilast ,\dlast)$ since the costs were not incremented. 

	If $\tau_S$  passes through $(j',d)$, then
	$c'(j',d)$ was incremented by $1$, while in $\cALG(j',d)$, it was not incremented at row $j'$ and only at row $i$. So the contribution of the mismatch with respect to both cost functions is the same. If $\tau_S$ was not traversing through $(j',d)$: At a later row $j''$, $\tau_S$ transits to diagonal $d$ and pays a cost $1$ for the transition. In $\cALG$ we pay for this transition twice: once at the transition row, and second time at row $i$ (note that each such a transition is counted only once). The claim follows.
	The analysis of the special diagonal $d^*$ follows by the same argument.
	
	If $\tau_S$ traverses other diagonals at row $i$ then with respect to $c'$ it pays a cost $1$ on each diagonal transition. However, on $\cALG$ we may pay either $0$ or $1$ cost. Therefore, the claim holds for all vertices at $\tau_S$ touching row $i$.  
\end{proof}

We conclude the proof by claiming that the algorithm outputs \close. By Corollary~\ref{cor:c'}, in $G'$ 
there exists a shortest path to the sink of cost $\le t/2$ with respect to the cost $c'$.
Therefore, by Claim~\ref{claim:cALGc'} there exists a shortest path to the sink of cost $\le t$ 
with respect to the cost $\cALG$. 
By Lemma~\ref{lem:cALG} after processing the last row of $S'$ the cost of $\cA[0]$ equals to the cost of the shortest path to 
$(i_{\abs{S'}},0)$  is  at most $t$. Therefore, by monotonicity the cost along $\cA[0]$ does not exceed $t$ and the algorithm outputs \close.
This completes the proof of Lemma~\ref{claim:close}. 
\end{proof}

\subsection{Proof of Far Case $\ed(x,y)> 13t^2$}	
	
	For the purposes of analysis, we think of the algorithm as first independently sampling two sets $S_1, S_2$, where each set $S_j, j=1,2$ is drawn such that each row is inserted into $S_j$ independently with probability $\frac {\log n}{t}$. While staying at the sampling mode, the algorithm uses $S_1$ for {\em periodicity check} and the set $S_2$ for {shift check}. 
		
	To conclude the proof we rely on the following claim, which is a variant of Claim~\ref{claim:badSamplingEvent}. Roughly speaking, we define bad events to capture scenarios in which the algorithm might fail, either that on a sampling round on a single diagonal we skip too many mismatches, or if we have one diagonal we skip too many period violations.
	
	\begin{claim}\label{claim:badSamplingEventGeneralized}
		Let $x,y \in \Sigma^n$ let $i\in [n]$, and let $S_j\subseteq [n]$, $j=1,2$ be sets drawn by including each row in $S$ independently with probability $\frac {\log n}{t}$.
		
		Let $g\le 2t+1, p\in \zo^g$ and $d\in \pmt$.
		Let: 
\ifprocs
\begin{align*}
  i_{d,i}&=\min_{r\ge i}{\Delta_H (x_{[i.. r]}, y_{[i.. r]+d})> 4t}, \\
         &I_{d,i}=\sett{j\in [i..i_{d,i}]}{x_j\neq y_{j+d}}; \\
  i_{x,p,i}&= \min_{r\ge i}{\Delta_H (x_{[i..r]}, p^*)> 4t}, \\
         &I_{x,p,i}=\sett{j \in [i..  i_{x,p,i}]}{x_j\neq p_{j-i \bmod p}}; \\
  i_{y,p,i}&= \min_{r\ge i}{\Delta_H (y_{[i..r]}, p^*)}> 4t, \\
         &I_{y,p,i}=\sett{j\in [i.. i_{y,p,i}]}{y_j\neq p_{j-i \bmod p}}; 
\end{align*}
\else
\begin{align*}
  i_{d,i}&=\min_{r\ge i}{\Delta_H (x_{[i.. r]}, y_{[i.. r]+d})> 4t}, 
         &I_{d,i}=\sett{j\in [i..i_{d,i}]}{x_j\neq y_{j+d}}; \\
  i_{x,p,i}&= \min_{r\ge i}{\Delta_H (x_{[i..r]}, p^*)> 4t}, 
         &I_{x,p,i}=\sett{j \in [i..  i_{x,p,i}]}{x_j\neq p_{j-i \bmod p}}; \\
  i_{y,p,i}&= \min_{r\ge i}{\Delta_H (y_{[i..r]}, p^*)}> 4t, 
         &I_{y,p,i}=\sett{j\in [i.. i_{y,p,i}]}{y_j\neq p_{j-i \bmod p}}; 
\end{align*}
\fi
		where $\Delta_H(x_{[i..r]}, p^*)$ is the Hamming distance between $x_{[i..r]}$ and the corresponding periodic string with period pattern $p$ of length $r-i$.
		If in either of the above definitions no such $r$ exists, we set $i_{d,i}, i_{x,p,i}, i_{y,p,i}=\infty$.
		
		We define an event $B_S(i,d)$ as the event where $i_{d,i}\neq \infty$ and $S\cap I_{d,i}=\emptyset$, similarly we define $B_S(i,x,p)$ as the event where $i_{x,p,i}\neq \infty$ and $S\cap I_{x,p,i}=\emptyset$ (and $B_S(i,y,p)$ is defined similarly).
		Then: $$\Pr [B_S(i,d)]<\frac{1}{9n(2t+1)}  $$ and similarly $\Pr [B_S(i,x,p)], \Pr [B_S(i,y,p)]$ are at most $\frac{1}{9n(2t+1)}$.
	\end{claim}

	The proof of Claim~\ref{claim:badSamplingEventGeneralized} follows immediately by Chernoff bound. Let us conclude the proof using the claim. First by union bound on the set of all possible rows and diagonals we get that except with probability $n(2t+1) \frac 1 {3n(2t+1)}\le \frac 1 9$, none of the events $B_{S_1}(i,d)$ happens. Moreover, using a union bound on the value of all possible $i\ge 2t+1$ and $j\in[0.. 2t+1]$ we have that that except with probability $2  ( n(2t+1) \frac 1 {9n(2t+1)})\le \frac 2 9$, none of the events $B_{S_2}(i,x_{[i-j.. i]},x), B_{S_2}(i,y_{[i-j.. i]},y)$ happens. So overall none of the events mentioned above happen with probability at least $2/3$.  We conclude the proof by showing that in such a case, the algorithm outputs $\far$ with probability $1$. 
	
	Assume for sake of contradiction that none of the events described above happen and still the algorithm outputs $\close$. We may view the algorithm as computing the shortest path $\tau$ on the grid graph $G_{S'}$, with respect to the cost function $\cALG$ defined on the proof of Claim~\ref{claim:close}.
	
	The path $\tau$ divides the set of rows processed by the algorithm into intervals $I'_0, \dots, I'_k\subseteq S$ such that (i) $I'_j$ and $I'_{j+1}$, $j \in [0,k-1]$ can intersect on at most single row, and (ii) for each interval $I'_{\ell}$, $\tau$ traverses along vertices on diagonal $d_\ell$ while paying cost $0$ (so moving between intervals correspond to either substitutions or diagonal transitions implying insertions and deletions).
	
	Observe that $k\le t-\abs {d_k}$ and $\max I_k$ is the largest round processed by the algorithm. 
	
	Let us define another set of intersecting intervals $I_0,\dots , I_k\subseteq [n]$ as follows:  $I_0=[0..\max I'_0]$,  $I_\ell=I'_\ell$ and $I_k=[\min I'_k.. n]$. 
	
	Consider the path $\tau$ in $G_{x,y}$ that on rows in $I_\ell$ follows diagonal $d_\ell$, while paying the cost along all diagonal edges, and then finally it traverses to $(n,0)$. 
	Let us denote by $\tau_{I_\ell}$ the cost of $\tau$ confined to the rows in $I_\ell$, which equals: $\Delta_H(x_{I_\ell}, y_{I_\ell+d_\ell})$. 
	
	\begin{claim}
		For all $\ell\in [k]$ we have:
		$$ c(\tau_{I_\ell})=\Delta_H(x_{I_\ell}, y_{I_\ell+d_\ell})<12t.$$ 
	\end{claim}
		
	\begin{proof}
		Let us consider an interval $I_\ell$: 
		Observe that once the diagonal $d_\ell$ is part of the diagonal on which the algorithm detects a mismatch followed by period transition (end of a sampling mode), then the current interval is over. Hence we can break each interval $I_\ell$ into its prefix $I_\ell^p$ that contains the rows on which the algorithm performs a shift check. And its suffix $I_\ell^s$ that contains the rows on which the algorithm performs periodicity check and finally encounters a mismatch on $d_\ell$.
		
		We conclude the claim by proving that  $\Delta_H(x_{I^p_\ell}, y_{I^p_\ell+d_\ell})<4t$ and $\Delta_H(x_{I^s_\ell}, y_{I^s_\ell+d_\ell})<8t$.

		Observe that since we assume that the event $B_{S_1}(\min I_\ell, d_\ell)$ did not happen then by definition $\Delta_H(x_{I^p_\ell}, y_{I^p_\ell+d_\ell})<4t$. Next since $B_{S_2}(i_\ell, x_{[i_\ell-g.. i-1]},x)$ and $B_{S_2}(i_\ell+\max \Diag, x_{[i_\ell-g.. i-1]},y)$ did not happen then we get: $\Delta_H(x_{[i_\ell +1.. \max I_\ell]}, p^*)<4t$ and $\Delta_H(y_{[i_\ell +1.. \max I_\ell]}+d_\ell, p^*)<4t$ where $p^*$ as usual denotes repeated concatenation of $p$. By triangle inequality we get that $\Delta_H(x_{[i_\ell +1..\max I_\ell]}, y_{[i_\ell +1..\max I_\ell]+d_\ell})<8t$, as claimed.
	\end{proof}
	
	In total this implies that the cost of $\tau$ is at most $k\times 12 t + k + t-\abs {d_k} \le 12t^2+t \le 13t^2 $, contradicting the fact that $\ed(x,y)>13t^2$.

\subsection{Query Complexity Analysis}

\begin{claim}
	The query complexity of the algorithm presented in Section~\ref{sec:AlgorithmDescription} is bounded by $\tilde{ O}(\frac{n}{t}+t^3)$.
\end{claim}
\begin{proof}
We bound the number of queries by showing that while staying at the contiguous mode the algorithm queries the input in $O(t^3\log{n})$ locations, and while staying on the sampling mode, it queries $O(\frac {n\log n}{t})$ locations with high probability.

Observe that the algorithm enters into the contiguous mode whenever it encounters a mismatch on at least one of the potent diagonals. Also observe that after $2(2t+1)$ rounds on which the algorithm stays at the contiguous mode, either there exists a mismatch on at least one of the potent diagonals or the algorithm shifts to the sampling mode. Therefore, for each $2(2t+1)$ consecutive rows, spent on the contiguous mode we can match a mismatch on at least one of the potent diagonals. Also observe that  (Lemma~\ref{lem:mismatchPotent}) each diagonal can have at most $O(t)$ mismatches while it is potent. Therefore, the number of $2(2t+1)$ consecutive rounds on  which the algorithm stays at the contiguous mode is bounded by $O(t^2)$.

Notice also that on each block $I$ of $2(2t+1)$ consecutive rounds on which the algorithm stays at the contiguous mode, it samples only $O(t)$ characters from both $x$ and $y$ (as it needs to compare from $x_I$ against values in $y_{\set{\min I-t, \dots, \max I+t}}$). Therefore, in total the number of queries used while staying at the contiguous mode is bounded by $O(t^3)$.

Regrading the sampling mode, at each sampling round, if both periodicity checks pass then the algorithm makes only $1$ queries into both strings $x,y$. So such rounds can contribute at most $\card{S}$ queries (which is bounded by $O(\frac {n\log n}{t})$ with high probability). Otherwise, it makes $O(t\log{n})$ additional queries and by Claim~\ref{claim:periodViolation} we are guaranteed to find at least one mismatch on one of the potent diagonals. Therefore, by the same argument used for the contiguous search the number of queries made by during the periodicity check of the sampling mode is bounded by $O(t^3\log{n})$. The number of rows sampled by the shift mode is bounded by $O(\frac {n\log n}{t})$ with high probability. Hence, the claim follows.
\end{proof}

\subsection{Time Complexity Analysis}
\begin{claim}
	The time complexity of the algorithm presented in Section~\ref{sec:AlgorithmDescription} is bounded by $\tilde{ O}(\frac n t+t^3)$.
\end{claim}
	
\begin{proof}
	Recall that the algorithm has two modes: the sampling mode and the contiguous mode.
	Let us first analyze the running time of the algorithm while staying at the contiguous mode.
	Naively, while staying at that mode, the algorithm has to pay at each round $O(\card {\Diag})=O(t)$ operations. Since there are at most $O(t^3)$ such rounds, then the total time spent at this mode is bounded by $O(t^4)$. However, we can utilize suffix tree machinery to accelerate the computation process.
	
	Observe that whenever the algorithm enters into the contiguous mode, then at each round $i$ and for each of the potent diagonals $d$ it has to update the cost based on whether $x_{i+1}=y_{i+d+1}$. Also observe that whenever the values match, the diagonal is retained as potent at the next round $i+1$. So if we have an efficient way to determine at row $i$ and for diagonal $d$, what is the maximal row $i_{max}\ge i$ such that: $x_{[i,\dots, i_{max}]}=y_{[i,\dots, i_{max}]+d}$. Using suffix trees machinery one can solve this problem in $O(1)$-time. However, building the suffix tree requires querying the entire strings $x,y$.
	
	Nevertheless, we may apply the suffix trees machinery only on substrings of $x,y$ of $5t$-length at a time. This ideas also have been implemented in see~\cite{BZ16,CGK16b} in the context of edit distance computation in a streaming fashion. 
	
	Using suffix trees, we pay $O(t)$-time generating the (truncated) suffix trees. Then, given a row $i$ and for diagonal $d$, using the suffix tree we can find the maximal row $i_{max}\in \set {i,\dots, i+5t}$ such that: $x_{[i,\dots, i_{max}]}=y_{[i,\dots, i_{max}]+d}$. For a given $i$, let $\Diag$ be the set of potent diagonals, and let $k_i$ be the mismatches encountered on the potent diagonals along the next $5t$-rows. Then the running time of the algorithm is bounded by $O(t+k_i^2)$. In total, if we divide the contiguous rounds into blocks of length $5t$ we get that the running time of the contiguous mode is bounded by $t^2(O(t))+\sum k_i^2  =O(t^3)$.

	While staying at the sampling mode at each round the algorithm performs only $O(1)$-operations, provided that the period was kept. This contributes $O(\frac {n\log n}t)$ factor to the running time. If the period was violated, then the algorithm first performs a binary search to detect a period transition, which requires $O(t\log n)$-time. Then it again can use suffix tree machinery to find out the list of diagonals having a mismatch in the next $2(2t+1)$ rows. This takes $O(t)$-time, and will be applied at most $O(t^2)$-times (as whenever this happens one of the potent diagonals increases its cost).
	
	Summarizing, using suffix tree machinery, the algorithm can be implemented so that its running time complexity is bounded by $\tilde{ O}(\frac n t+t^3)$, as claimed.
	\end{proof}

\subsection{Distinguishing $t$ vs. $O(t^{2-\epsilon})$-gap in time $\tilde O(\frac{n}{t^{1-\epsilon}}+t^3)$}

\begin{proposition} 
	There exists an algorithm that, 
	given as input strings $x,y\in \zo^n$ and an integer $t \le \sqrt{n}$,
	has query and time complexity bounded by $\tilde O(\frac{n}{t^{1-\epsilon}}+t^3)$,
	and satisfies the following:
	\begin{itemize}
		\item If $\ed(x,y)\le t/2$ it outputs \close with probability $1$.
		\item If $\ed(x,y)= \Omega (t^{2-\epsilon})$ it outputs \far with probability at least  $2/3$.
	\end{itemize}
\end{proposition}

Let us briefly sketch the proof. We keep the same algorithm structure, where the only change is in the rate of sampling: instead of sampling at a rate $\tilde{O} (\frac{1}{t})$, our algorithm samples at a rate $\tilde{ O} (\frac{1}{t^{1-\epsilon}})$. That will increase the query and running complexity of the modified algorithm into $\tilde{ O}(\frac{n}{t^{1-\epsilon}}+t^3)$. 

The proof of correctness follows by the same arguments, where one have to show that while sampling at this rate with high probability the algorithm detects mismatches and period violations at a rate of $\frac{1}{t^{1-\epsilon}}$. 

\subsection{Succinct Representation of an Alignment}
\label{sec:succinct}
 Our algorithm can succinctly represent an alignment in $\tilde{O}(t^2)$ bits. Note that our algorithm has at most $O(t^2)$ sampling mode, therefore at most $O(t^2)$ contiguous mode. During the $i$th contiguous mode, we can represent an alignment  using $\tilde{O}(k_i)$ bits if $k_i$ is the number of edit operations paid during that contiguous mode. Since, in each sampling mode, the alignment is given by a single diagonal, it can be represented in $\tilde{O}(1)$ bits. Thus, over all the contiguous and sampling modes, representing the alignment requires $\tilde{O}(\sum_{i} k_i +t^2)=\tilde{O}(t^2)$ bits.

\ifprocs
\else {
  \small
  \bibliographystyle{alphaurlinit}
}
\fi
  
\bibliography{editDistance}

\appendix	

\section{Proof of Lemma~\ref{lem:gridGraph}}
\label{app:gridGarphProof}

\ifprocs
\section*{A. Proof of Lemma~\ref{lem:gridGraph}}
\fi

\begin{lemma*}[\ref{lem:gridGraph}]
	Let $\ed(x,y)\in \zo^n$.
	
	If $\ed(x,y)\le t$ then with probability $1$
	the algorithm outputs \close. 
	
	If $\ed(x,y)>6 t^2$ then 
	with probability at least $2/3$ the algorithm outputs \far.
\end{lemma*}

\begin{proof}[Proof of Lemma~\ref{lem:gridGraph}]
To prove the first part, suppose $\ed(x,y)\le t$.
We will prove for all $S\subseteq [n]$ (and thus with probability $1$),
for every source-to-sink path $\tau$ in the original grid graph $G_{x,y}$,
there is in $G_S$ a corresponding source-to-sink path $\tau_S$ of the same or lower cost.
It would then follow that $G_S$ contains a path from the source $(0,0)$
to the sink $(n,0)$ of cost at most $t$, and the algorithm outputs \close.

Given $S\subseteq [n]$ and a path $\tau$ in $G_{x,y}$, 
construct the corresponding path $\tau_S$ in $G_S$ as follows.
Suppose the vertices $\tau$ traverses in row $0$ are $(0,0),\dots,(0,d_0)$;
then let $\tau_S$ start at $(0,0)$ and traverse the exact same vertices.
Now we describe how to extend $\tau_S$ iteratively for $j=0,\ldots,s-1$,
where $i_0=0$ by convention. 
Denote the last vertex $\tau_S$ traverses in row $i_j$ by $(i_j,d_S)$, 
and suppose the vertices $\tau$ traverses in row $i_{j+1}$
are $(i_{j+1},d),\dots,(i_{j+1},d+\ell)$. 
Now we have two cases:
if $d_S\le d+\ell$,
extend $\tau_S$ by appending $(i_{j+1},d_S) \dots, (i_{j+1},d+\ell)$;
otherwise, extend it by appending $(i_{j+1},d_S-1)$.
Finally, denote the last vertex $\tau_S$ traverses in row $i_s$ by $(i_s,d_S)$;
then extend $\tau_S$ by append $(n,0)$, which uses an edge of cost $\abs{d_S}$. 

\begin{claim}\label{claim:closeCaseAnalysis}
$c_{G_S}(\tau_S) \le c_{G_{x,y}}(\tau)$.  
\end{claim}  

\begin{proof}[Proof of Claim~\ref{claim:closeCaseAnalysis}]
For each $j=0,\ldots,s$,
let $d_\tau(j)$ denote the last diagonal visited by $\tau$ at row $i_j$,
and let $d_{\tau_S}(j)$ be similarly for the path $\tau_S$.
Denote by $c_\tau(j)$ the cost of the prefix of $\tau$ up to $(i_j,d_\tau (j))$, and similarly $c_{\tau_S}(j)$ for path $\tau_S$ and vertex $(i_j,d_{\tau_S}(j))$.
We will prove show the following bound on $\tau_S$
\begin{equation} \label{eq:close1}
  \forall j=0,\ldots,s, \qquad
  c_{\tau_S}(j) + d_{\tau_S} (j)
  \le c_\tau (j) + d_\tau(j). 
\end{equation}
Let us now show how this bound implies the claim. 
By construction, the last edge in $\tau_S$ goes from row $i_s$
(the last row in $S$) to the sink and has cost $\abs{d_{\tau_S}(s)}$,
and together with~\eqref{eq:close1} in the case $j=s$, we have
\begin{align*}
  c_{G_S}(\tau_S)
  & = c_{\tau_S} (s) + \abs{d_{\tau_S}(s)}
  \\
  & \le c_\tau (s) + d_\tau(s) - d_{\tau_S} (s) + \abs{d_{\tau_S}(s)} .
\intertext{Now if $d_{\tau_S} (s) \ge0$, the last two summands above cancel and we continue}
  & = c_\tau (s) + d_\tau(s) + 0 
  \le c_\tau (s) + \abs{d_\tau(s)}; 
\intertext{otherwise, we have $d_{\tau} (s) \le d_{\tau_S} (s) < 0$ and we continue}
  & \le c_\tau (s) + 0 + \abs{d_{\tau_S}(s)}  
  \le c_\tau (s) + \abs{d_\tau(s)} . 
\end{align*}
In both cases, we obtain 
$
  c_{G_S}(\tau_S)
  \le c_\tau (s) + \abs{d_\tau(s)}  
  \le c_{G_{x,y}}(\tau)
$, which proves the claim. 

We proceed to proving \eqref{eq:close1} by induction on $j$.
The base case $j=0$ holds trivially (with equality), 
because $\tau_S$ is constructed to be identical to $\tau$ in row $i_0=0$.
For the inductive step, we actually show
\begin{equation} \label{eq:close2}
  \forall j=0,\ldots,s-1, \qquad
  \Delta_j c_{\tau_S} + \Delta_j d_{\tau_S}
  \le \Delta_j c_\tau + \Delta_j d_\tau ,
\end{equation}
where $\Delta_j f := f(j+1) - f(j)$
for $f(j)$ being any of the 4 terms appearing in~\eqref{eq:close1}.  
The last inequality clearly implies the inductive step. 
(Alternatively, we can replace the induction by a telescopic sum.)

Now to prove~\eqref{eq:close2}, fix $j\in\set{0,\ldots,s-1}$,
and observe that in the desired inequality,
the LHS is about the subpath of $\tau_S$ from row $i_j$ to row $i_{j+1}$,
and similarly the RHS is about the subpath of $\tau$.
More precisely, these subpaths are taken to ``start'' and ``end'' 
at the last vertex visited in each row.
Assume first that $LHS = 0$.
Observe that $\Delta_j c_\tau$ is the cost along that subpath of $\tau$,
and every edge in it that increments/decrements the diagonal has cost $1$,
hence $\Delta_j c_\tau\ge \abs{\Delta_j d_\tau} \ge - \abs{\Delta_j d_\tau}$. 
This proves that in this case, indeed $RHS \ge 0 = LHS$.

Assume next that $LHS > 0$.
Suppose towards contradiction that the first edge in that subpath of $\tau_S$
(from row $i_j$ to row $i_{j+1}$) decreases the diagonal; 
then by its construction, $\tau_S$ visits no additional vertices on this row,
hence the said subpath of $\tau_S$ consists of only one edge,
and we see that $LHS = \Delta_j c_{\tau_S} + \Delta_j d_{\tau_S} = 1 - 1 =0$,
which contradicts our assumption.
We thus know that the first edge in that subpath of $\tau_S$ 
does not change the diagonal. %
Clearly, any additional edges in this subpath, if any, 
must stay in the same row and increment the diagonal,
hence their number is exactly $\Delta_j d_{\tau_S} \ge 0$.
Observe the fact $d_{\tau_S}(j) \ge d_{\tau}(j)$,
which follows from the construction of $\tau_S$. 
We can also verify the fact $d_{\tau_S}(j+1) = d_{\tau}(j+1)$;
indeed, one direction ($\ge$) is just the previous inequality (but for $j+1$), 
and the other direction ($\le$) holds in our case
where the first edge of the subpath does not change the diagonal. 
Combining these two facts, we have 
$\Delta_j d_{\tau_S} \le \Delta_j d_{\tau}$.
Moreover, the foregoing discussion implies that
\begin{equation} \label{eq:close3}
  0\le \Delta_j d_{\tau_S} \le \Delta_j d_{\tau} \le \Delta_j c_{\tau}.
\end{equation}
Observe that by the foregoing discussion and the definition of $G_S$, 
\begin{equation} \label{eq:close5}
  \Delta_j c_{\tau_S} 
  = \indicator{\set{x_{i_{j}+1}\neq y_{i_{j}+d_{\tau_S}(j)+1}}} + \Delta_j d_{\tau_S} 
  \le 1 + \Delta_j d_{\tau_S} ,
\end{equation}
and let us argue next, by a case analysis, that 
\begin{equation} \label{eq:close4}
  \Delta_j c_{\tau_S} 
  \le \Delta_j c_{\tau}, 
\end{equation}
Case 1 of proving~\eqref{eq:close4} is when $d_\tau(j) < d_{\tau_S}(j)$
(i.e., our first fact above holds with strict inequality).
Then the derivation of~\eqref{eq:close3} actually gives a stronger bound
$\Delta_j d_{\tau_S} + 1 \le \Delta_j d_{\tau}$.
Combining this with~\eqref{eq:close3} and~\eqref{eq:close5} 
we obtain~\eqref{eq:close4}. 

Case 2 %
is when the first edge in that subpath of $\tau$ decrements the diagonal. 
Then later steps in the subpath must increment the diagonal
(because the net difference is $\Delta_j(\tau) \ge 0$), 
and again we obtain a stronger cost bound
$\Delta_j c_\tau \ge \Delta_j d_{\tau} + 2$. 
Combining this with~\eqref{eq:close3} and~\eqref{eq:close5} 
we obtain~\eqref{eq:close4}. 

Case 3 %
is the remaining scenario,
where $d_\tau(j) = d_{\tau_S}(j)$
and the first edge in that subpath of $\tau$ does not change the diagonal,
and thus has cost
$\indicator{\set{x_{i_{j}+1}\neq y_{i_{j}+d_\tau(j) +1}}}$. 
Hence, 
\[
  \Delta_j c_\tau
  \ge \indicator{\set{x_{i_{j}+1}\neq y_{i_{j}+d_\tau(j) +1}}} + \Delta_j d_{\tau}. 
\]
Combining this with~\eqref{eq:close5} implies~\eqref{eq:close4}.

Finally, having established~\eqref{eq:close4},
we combine it with~\eqref{eq:close3} to derive~\eqref{eq:close2},
which we called $LHS \leq RHS$, and proves the case $LHS > 0$.
This completes the proof of the inductive step
and of Claim~\ref{claim:closeCaseAnalysis}.
\end{proof}

This completes the proof of the first part of Lemma~\ref{lem:gridGraph}.

To prove the second part, suppose that $\ed(x,y)>6 t^2$.
Using Claim~\ref{claim:badSamplingEvent} and applying a union bound
on all possible rows and diagonals, we get that except with probability $n(2t+1) \frac 1 {3n(2t+1)}\le \frac 1 3$, none of the events $B(i,d)$ happens.
We conclude the proof by showing that in this case, the shortest path connecting $(0,0)$ and $(n,0)$ has cost strictly larger than $t$,
and therefore our algorithm outputs \far.
	
Assume towards contradiction that none of the events $B(i,d)$ happens
and yet the shortest path in $G_S$ from $(0,0)$ to $(n,0)$, denoted $\tau_S$,
has cost at most $t$.
From $\tau_S$, we construct a new path $\tau$ in $G_{x,y}$ as follows. (i) For each edge $(i_j,d)$ to $(i_j,d+1)$ in $\tau_S$, we include the same edge in $\tau$. (ii) For each edge $(i_j,d)$ to $(i_{j+1},d-1)$, we include the edges corresponding to the path $(i_j,d), (i_j+1,d),\ldots,(i_{j+1}-1,d)$ followed by an edge to $(i_{j+1},d-1)$. (iii) For each edge $(i_j,d)$ to $(i_{j+1},d)$, we include the edges corresponding to the path $(i_j,d), (i_j+1,d),\ldots,(i_{j+1},d)$. 
	
Notice that in case (i), $\tau_S$ pays a cost of $1$ and so does $\tau$. In case (ii), $\tau_S$ pays a cost of $1$ and since $B(i_j,d)$ did not happen, $\tau$ pays a cost of at most $3t+(d'-d)$.
	
	Now consider the maximal subpaths that are formed in $\tau$ by the edges in case (iii). Each of these maximal subpaths can be indexed by a contiguous collection of rows and a single diagonal. Let us denote them by $(I_1,d_1)\dots , (I_k,d_k)$.
	Let $I'_j \subseteq I_j$ be the rows in $I_j \cap S$ for $j=1,\ldots,k$.
	
	Therefore in $\tau_S$, there is a path through diagonal $d_j$ and rows in $I'_j$. Let this path pay $e_j$ edit cost in $\tau_S$. Note that they are all from substitution edits. Let $Z=\{z_1, z_2,\ldots,z_{e_j}\}$ be the rows in $I'_j$ in increasing order such that $\tau_S$ pays an edit cost on the outgoing edge from $(z,d_j)$ for all $z \in Z$. 
	Since, we avoided the events $B(\min I'_j,d_j), B(z_1+2,d_j),\dots, B(z_{e_j}+2,d_j)$, the total substitution cost paid in $\tau$ while traversing through $(I_j,d_j)$ is at most $3t(e_j+1)$. 
	
	Therefore, over all $k$, $\tau$ pays a total substitution cost of $\sum_{j=1}^{k}3t(e_j+1) \leq 3t^2+3t \sum_{j=1}^{j} e_j$, since $k \leq t$. Now adding the cost from case (i) and (ii), the overall cost paid by $\tau$ is at most $6t^2$. This contradicts the assumption that $\ed(x,y)>6t^2$.

This completes the proof of Lemma~\ref{lem:gridGraph} (both parts). 
\end{proof}

\section{Missing Proofs from Section~\ref{sec:subLinear}} 
\label{app:appendixPotential}
	
\ifprocs
\section*{B. Missing Proofs from Section~\ref{sec:subLinear}} 
\fi

\begin{proof}[\textbf{Proof of Lemma~\ref{lem:potential}}] 
We proceed by induction on the grid vertices in lexicographic order
(i.e., their row is the primary key and their diagonal is secondary). 
The base case is the source $(0,0)$, which follows trivially.
To prove the inductive step,
consider a vertex $v$ and assume the claim holds for all previous vertices. 
We now split into two cases. 

Case (a): $v$ is dominated. 
By applying the induction hypothesis to the in-neighbor that dominates $v$,
we obtain a shortest path $(0,0)=v_0,v_1,\ldots,v_l$
to the in-neighbor $v_l$ that dominates $v$,
hence the cost of this path is $c(v_l) \leq c(v)-1$. 
By appending that path with $v$, we obtain a shortest path to $v$
because its cost is at most $c(v_l)+1 \leq c(v)$.
It remains to show the ordering requirement
that all non-potent vertices in this path appear after all potent vertices. 
If $v$ is non-potent, this is immediate because $v$ is appended. 
If $v$ is potent then by Definition~\ref{def:potent}
the dominating vertex $v_l$ must be potent, 
and again the ordering requirement follows immediately.

Case (b): $v=(i,d)$ is not dominated, and in particular it is potent.
Then a shortest-path to $v$ must be coming from $w_1=(i-1,d)$.
If this $w_1$ is potent, 
then we can obtain a shortest path to $w_1$ by the induction hypothesis,
and appending $v$ to this path satisfies the ordering requirement
and gives a shortest path because its cost is at most $c(w_1)+1\le c(v)$. 
So assume henceforth that $w_1$ is non-potent,
and let us show a contradiction. 
Then by Definition~\ref{def:potent} it must be dominated by some in-neighbor $w_2$, 
and moreover either $w_2$ is non-potent
or it has an outgoing edge of cost $0$, i.e., a matching edge (or both). 
If $w_2$ is non-potent, then by the same argument (as for $w_1$), 
it must be dominated by some in-neighbor $w_3$.
Repeat this argument until reaching the first $w_l$, $l\ge 2$, that is potent,
and thus $w_l$ must have an outgoing edge of cost $0$.
The path's construction and the monotonicity property~\eqref{eq:mono}
imply that $c(v)\ge c(w_1) \ge c(w_2) - 1\ge \cdots \ge c(w_l)-(l-1)$. 
Observe that the path $w_l\to\cdots\to w_2\to w_1\to v$ has $l$ edges,
the first $l-1$ are insertion/deletion edges,
and the last one is a matching or substitution edge. 
Consider an alternative path from $w_l$,
that uses first the outgoing edge of cost $0$,
and then $l-1$ insertion/deletion edges in exact correspondence
with the $l-1$ edges $w_l\to w_{l-1}, \ldots, w_2\to w_1$. 
It is easy to verify that also this alternative path reaches $v$.
and its cost is exactly $l-1$.
Appending this path to an arbitrary shortest-path to $w_l$, 
yields a path of cost $c(w_l) + l-1 \le c(v)$,
i.e., a shortest-path to $v$ that enters $v$ via an insertion/deletion edge.
This means that $v$ is dominated, in contradiction to our assumption.
\end{proof}

\medskip
\begin{proof}[\textbf{Proof of Lemma~\ref{lem:imp}}] 
First, if $c((i,d) \rightarrow (i, d+1))=0$, then $c(i+1, d) \leq c(i,d)+c((i,d) \rightarrow (i, d+1))=c(i,d)$. On the other hand, from the monotonicity property~\eqref{eq:mono}, $c(i+1,d) \geq c(i,d)$. Therefore, whenever $c((i,d) \rightarrow (i, d+1))=0$, we have $c(i,d)=c(i,d+1)$.

Therefore, let us assume, $c((i,d) \rightarrow (i, d+1))=1$. We prove the claim by a double induction on the grid vertices $(i,d)$ in lexicographic order
(i.e., their row is the primary key and their diagonal is secondary).

Base case: Note that by definition $(0,0)$ is potent. Let $(i,d)$ be the {\em first} vertex in the lexicography order which is non-potent. Then $(i,d)$ must be dominated by either $(i,d-1)$ or $(i-1,d+1)$ both of which are potent (if they exist). Let w.l.o.g., $(i,d-1)$ dominates $(i,d)$, then if $c((i,d-1) \rightarrow (i+1,d-1))=1$, it must hold for $(i,d)$ to be non-potent that $(i-1,d+1)$ dominates $(i,d)$ and has $c((i-1,d+1) \rightarrow (i,d+1))=0$. Thus consider $(i,d-1)$ dominates $(i,d)$ and has $c((i,d-1) \rightarrow (i+1,d-1))=0$ (the other case is identical). Then $c(i+1,d-1)=c(i,d-1)$. Overall, we have
\ifprocs
\begin{align*}
  c(i,d)
  \leq c(i+1,d)
  & \leq c(i+1,d-1)+1
    \\ &
    = c(i,d-1)+1=c(i,d),
\end{align*}
\else
\begin{align*}
  c(i,d)
  \leq c(i+1,d)
  & \leq c(i+1,d-1)+1
    = c(i,d-1)+1=c(i,d),
\end{align*}
\fi
where the last equality follows from $(i,d-1)$ dominating $(i,d)$. Therefore, $c(i+1,d)=c(i,d)$.

For the inductive step, assume the claim is true for all vertices $(i',d')$ with $i' <i$ and consider row $i$. Let $d$ be the \emph{first} diagonal on row $i$ such that $(i,d)$ is non potent. Since $(i,d)$ is non-potent, it must be dominated by either $(i,d-1)$ or $(i-1,d+1)$. If $(i,d)$ is only dominated by $(i,d-1)$, then since $(i,d-1)$ is potent (note $d$ is the first diagonal on row $i$ which is non-potent), then for $(i,d)$ to be non-potent, it must hold $c((i,d-1) \rightarrow (i+1,d-1))=0$. Now following the same argument as in the base case, we get $c(i+1,d)=c(i,d)$.

Thus assume, either $(i,d)$ is not dominated by $(i,d-1)$ or $c((i,d-1) \rightarrow (i+1,d-1))=1$. Then, $(i-1,d+1)$ dominates $(i,d)$, that is $c(i-1,d+1)=c(i,d)-1$. If $(i-1,d+1)$ is potent, then for $(i,d)$ to be non-potent, we must have $c((i-1,d+1) \rightarrow (i,d+1))=0$. This implies $c(i,d+1)=c(i-1,d+1)$. Otherwise, $(i-1,d+1)$ dominates $(i,d)$ but $(i-1,d+1)$ is non-potent. Then from the induction hypothesis, $c(i,d+1)=c(i-1,d+1)$. Overall, we have
\ifprocs
\begin{align*}
  c(i,d)
  \leq c(i+1,d)
  & \leq c(i,d+1)+1
    \\ &
    =c(i-1,d+1)+1
    =c(i,d).
\end{align*}
\else
\begin{align*}
  c(i,d)
  \leq c(i+1,d)
  & \leq c(i,d+1)+1
    =c(i-1,d+1)+1
    =c(i,d).
\end{align*}
\fi
Therefore, $c(i+1,d)=c(i,d)$. Thus, the claim holds for the first non-potent diagonal on row $i$. Suppose, again by the inductive hypothesis, the claim holds for all $r'$th potent diagonal on row $i$ with $r' <r$, and consider the $r$-th non-potent diagonal $d_r$ at row $i$, $r > 1$.

If $(i,d_{r}-1)$ is potent, then the claim follows from the same argument as above, as if $d_r$ is the first diagonal on row $i$ to be non-potent. 

Otherwise, $(i, d_{r}-1)$ is non-potent. Then, by the induction hypothesis, $c(i+1,d_{r}-1)=c(i,d_{r}-1)$. If $(i,d_r-1)$ dominates $(i,d_r)$, then we have
\ifprocs
\begin{align*}
  c(i,d_r)
  \leq c(i+1, d_r)
  & \leq c(i+1, d_r-1)+1
    \\ &
    = c(i,d_r-1)+1
    = c(i,d_r).
\end{align*}
\else
\begin{align*}
  c(i,d_r)
  \leq c(i+1, d_r)
  & \leq c(i+1, d_r-1)+1
    = c(i,d_r-1)+1
    = c(i,d_r).
\end{align*}
\fi
Therefore, $c(i+1, d_r)=c(i,d_r)$.

Otherwise, $(i,d_{r}-1)$ does not dominate $(i,d_r)$. Hence $(i,d_r)$ must be dominated by $(i-1,d_{r}+1)$. In that case, if $(i-1,d_{r}+1)$ is potent, then we must have $c( (i-1, d_{r}+1) \rightarrow (i, d_r+1)=0$. This implies $c(i,d_{r}+1)=c(i-1,d_{r}+1)$. On the other hand, if $(i-1, d_{r}+1)$ is non-potent, then from the induction hypothesis $c(i,d_{r}+1)=c(i-1,d_{r}+1)$. We have
\ifprocs
\begin{align*}
  c(i,d_r)
  \leq c(i+1,d_r)
  & \leq c(i,d_r+1)+1
    \\ &
    = c(i-1,d_r+1)+1
    = c(i,d_r).
\end{align*}
\else
\begin{align*}
  c(i,d_r)
  \leq c(i+1,d_r)
  & \leq c(i,d_r+1)+1
    = c(i-1,d_r+1)+1
    = c(i,d_r).
\end{align*}
\fi
Therefore, we have $c(i+1,d_r)=c(i,d_r)$. The lemma follows.
\end{proof}

\medskip
\begin{proof}[\textbf{Proof of Lemma~\ref{lem:mismatchPotent}}] 
We prove this by induction on the grid vertices $(i,d)$ in lexicographic order
(i.e., their row is the primary key and their diagonal is secondary). 
The base case $(i,d)=(0,0)$ is immediate, 
because the mismatch
guarantees that $c(1,0)\ge 1$. 

For the inductive step, 
consider $(i,d)$ and assume the claim holds for all previous vertices. 
Suppose first that $(i+1,d)$ has a mismatch,
and let us show that $c(i+1,d)=c(i,d)+1$.
By~\eqref{eq:mono}, it suffices to show that $c(i+1,d) \ge c(i,d)+1$,
which due to the mismatch requires proving a lower bound
on the cost of two in-neighbors of $(i+1,d)$, namely,
\[
  \minn{ c(i+1,d-1), c(i,d+1) } \geq c(i,d).
\]

This is equivalent to two inequalities,
and we start with the inequality
$
c(i+1,d-1)
\ge c(i,d)
$.
Assume first that $(i,d)$ is not dominated by $(i,d-1)$. 
This together with the monotonicity property~\eqref{eq:mono} implies
$
c(i,d)
\le c(i,d-1)
\le c(i+1,d-1)
$, 
as required.
Assume next that $(i,d)$ is dominated by $(i,d-1)$.
Then by Definition~\ref{def:potent}, 
diagonal $d-1$ is potent at row $i$ and has a mismatch at the next row $i+1$.
Applying the induction hypothesis to $(i,d-1)$
and using the bounded difference property~\eqref{eq:bdd_diff1}, we have
$
c(i+1,d-1)
= c(i,d-1) + 1
\ge c(i,d)
$, 
as required.
Thus, both cases satisfy the required inequality. 

The second inequality $c(i,d+1) \ge c(i,d)$ is proved by a similar argument 
with two cases depending on whether $(i,d)$ is dominated by $(i-1,d+1)$.
This concludes the inductive step in this case.

Suppose next that $(i+1,d)$ has a match.
Then this vertex has an incoming edge of cost $0$,
hence $c(i+1,d) \le c(i,d)$.
The other direction $c(i+1,d) \ge c(i,d)$ follows by~\eqref{eq:mono}. 
It remains to prove that $(i+1,d)$ is potent. First, assume $(i,d)$ is not dominated by $(i,d-1)$. 
This along with the monotonicity property~\eqref{eq:mono} implies $c(i+1,d-1)\geq c(i,d-1) \geq c(i,d)=c(i+1,d)$, i.e. $(i+1,d-1)$ does not dominate $(i+1,d)$. 

Otherwise, assume $(i,d-1)$ dominates $(i,d)$, then from Definition~\ref{def:potent}, $(i,d-1)$ must be potent and there must be a mismatch at row $i$. Hence, from the first part of this lemma, $c(i+1,d-1)=c(i,d-1)+1 \geq c(i,d)=c(i+1,d)$, i.e., $(i+1,d-1)$ cannot dominate $(i+1,d)$. 

Similarly, $(i,d+1)$ cannot dominate $(i,d)$. Therefore, $(i+1,d)$ is potent.
This completes the inductive step and proves the lemma. 
\end{proof}

\medskip
\begin{proof}[\textbf{Proof of Lemma~\ref{lem:ComputingPotential}}] 
We prove the first claim (about insertion to $\Diag_i$)
by induction on the grid vertices $(i,d)$ in lexicographic order
(i.e., their row is the primary key and their diagonal is secondary).
The base case is row $i=0$, at which the only potent diagonal is $d=0$,
and indeed it is inserted into $\Diag_0$ as initialization. 

For the inductive step, consider a potent vertex $(i,d)$
and assume the claim holds for all previous vertices.
By Lemma~\ref{lem:potential}, there is a shortest path to $(i,d)$
consisting only of potent vertices.
This path enters $(i,d)$ from one of its three in-neighbors
$(i-1,d)$, $(i-1,d+1)$, and $(i,d-1)$, which then must be potent too. 
We now have three cases.

Suppose first the shortest path enters from $(i-1,d)$.
As mentioned above, this vertex must be potent,
and then by the induction hypothesis, $d$ is inserted to the list $\Diag_{i-1}$.
It follows that when $(i-1,d)$ is scanned,
the algorithm will see it is potent and insert $d$ to $\Diag_i$. 

Second, assume the shortest path enters from $(i,d-1)$.
As mentioned above, it must be potent,
and then by the induction hypothesis, $d-1$ is inserted to the list $\Diag_i$.
Observe that $(i,d)$, which is potent,
must be dominated by $(i,d-1)$ because of the shortest path,
hence diagonal $d-1$ has a mismatch at the next row $i+1$,
and thus when $(i,d-1)$ is scanned,
the algorithm will insert $d$ to $\Diag_{i}$.

Third, assume the shortest path enters from $(i-1,d+1)$.
Similarly to the previous case, this vertex must be potent
and then by the induction hypothesis, $d+1$ is inserted to the list $\Diag_{i-1}$.
Since $(i,d)$ is potent and dominated by $(i-1,d+1)$,
and diagonal $d+1$ must have a mismatch at the next row $i$,
and thus when $(i-1,d+1)$ is scanned, 
the algorithm will insert $d$ to $\Diag_{i}$.

Finally, to prove the second claim, observe that during the scan of $\Diag_i$, 
every diagonal $d$ that is found to be not potent is removed from the list.
\end{proof}

\medskip
\begin{proof}[\textbf{Proof of Lemma~\ref{lem:costCorrectness}}]
We prove only the first assertion,
as the second one is an immediate consequence of it.
The proof is by induction on the grid vertices $(i,d)$ in lexicographic order
(i.e., their row is the primary key and their diagonal is secondary).
The base case is the time before processing vertex $(0,0)$; 
at this time, $\cA$ stores its initial values, i.e., $\cA[d] = \abs{d}$,
which is equal to $c(0,d)$ for all $d\ge 0$. 
For $d<0$, the base case is the time before processing $(-d,d)$,
because we should only consider vertices reachable from $(0,0)$;
at this time, $\cA[d]=\abs{d}$ is still the initialized value
and it is equal to $c(d,-d)=-d$. 

For the inductive step, we need to show that
the processing of diagonal $d\in\Diag_i$ updates the array $\cA$ correctly.
But using the induction hypothesis, 
we only need to show $\cA[d]$ is updated from $c(i,d)$ to $c(i+1,d)$.
To this end, suppose first that vertex $(i,d)$ is non-potent. 
Then by Lemma~\ref{lem:imp} we have $c(i+1,d)=c(i,d)$,
and the algorithm indeed does not modify $\cA[d]$. 
Suppose next $(i,d)$ is potent and let us use Lemma~\ref{lem:mismatchPotent}:
If $(i+1,d)$ has a mismatch then $c(i+1,d)=c(i,d)+1$,
and the algorithm indeed increments $\cA[d]$ by $1$;
and if $(i+1,d)$ has a match then $c(i+1,d)=c(i,d)$,
and the algorithm indeed does not change $\cA$ at all. 
\end{proof}

\end{document}

%% file: macros.tex
\usepackage[]{algorithm2e}

\def\compactify{\itemsep=0pt \topsep=0pt \partopsep=0pt \parsep=0pt}

\newcommand{\eqdef}{:=}
\newcommand{\Diag}{\mathcal{D}}

\newcommand{\set}[1]{\{#1\}}
\newcommand{\card}[1]{\left|#1\right|}

\newcommand{\ed}{\Delta_e}
\newcommand{\ham}{\Delta_H}

\newcommand{\abs}[1]{{\left|#1\right|}}
\newcommand{\minn}[1]{\min\{{#1}\}}
\newcommand{\maxx}[1]{\max\{{#1}\}}
\newcommand{\zo}{{\set{0,1}}}

\newcommand{\pmt}{{[-t..t]}}
\newcommand{\pmn}{{[-n..n]}}

\newcommand{\close}{\texttt{close}\xspace}
\newcommand{\far}{\texttt{far}\xspace}

\newcommand{\tO}{\tilde O}

\newcommand\sett[2]{\left\{  {#1}:\ {#2} \right\}}

\newcommand{\N}{\mathbb{N}}

\newcommand{\ilast}{i_\textrm{last}} 
\newcommand{\cALG}{c_\textrm{ALG}} 
\newcommand{\cA}{c_\textrm{A}} 
\newcommand{\Dnext}{\Diag_\mathrm{next}}
\newcommand{\ipat}{i_\mathrm{pat}}
\newcommand{\inext}{i_\mathrm{next}}
\newcommand{\dlast}{d_\mathrm{last}}

\theoremstyle{plain}
\newtheorem{proposition}[theorem]{Proposition}

\makeatother